\documentclass{vldb}
\usepackage{times}
\usepackage{graphicx}
\usepackage{balance}
\usepackage{pseudocode}
\usepackage{url}

\usepackage{times}

\newtheorem{definition}{Definition}
\newtheorem{theorem}{Theorem}

\begin{document}
\title{Distributed Data Placement via Graph Partitioning}

\numberofauthors{3}

\author{
\alignauthor
Lukasz Golab \\
\affaddr{University of Waterloo} \\
\email{lgolab@uwaterloo.ca}
\alignauthor
Marios Hadjieleftheriou \\
\affaddr{AT\&T Labs - Research} \\
\email{marioh@research.att.com}
\alignauthor
Howard Karloff \\
\affaddr{Yahoo! Labs} \\
\email{karloff@yahoo-inc.com}
\and
\alignauthor
Barna Saha \\
\affaddr{AT\&T Labs - Research} \\
\email{barna@research.att.com}
}

\maketitle

\sloppy

\begin{abstract}

With the widespread use of shared-nothing clusters of servers, there has been a proliferation of distributed object stores that offer high availability, reliability and enhanced performance for MapReduce-style workloads.  However, relational workloads 
cannot always be evaluated efficiently using MapReduce without extensive data migrations, which cause network congestion and reduced query throughput.  We study the problem of computing data placement strategies
that minimize the data communication costs incurred by typical relational
query workloads in a distributed setting.

Our main contribution is a reduction of the data placement problem to the well-studied problem of {\sc Graph Partitioning}, which is NP-Hard but for which efficient approximation algorithms exist.  
The novelty and significance of this result lie in representing
the communication cost exactly and using standard graphs instead of
hypergraphs, which were used in prior work on data 
placement that optimized for different objectives (not communication cost).

We study several practical extensions of the problem: with load balancing, with replication, with materialized views, and with complex query plans consisting of sequences of intermediate operations that may be computed on different servers.
We provide integer linear programs (IPs) that may be used with any IP solver to find an optimal data placement.  For the no-replication case, we use publicly available graph partitioning libraries (e.g., METIS) to efficiently compute nearly-optimal solutions.
For the versions with replication, we introduce two heuristics that utilize the {\sc Graph Partitioning} solution of the no-replication case.  Using the TPC-DS workload, it may take an IP solver weeks to compute an optimal data placement, whereas our reduction produces nearly-optimal solutions in seconds.
\end{abstract}

\section{Introduction}

The emergence of cloud computing has led to a proliferation of distributed storage solutions
that offer high availability, reliability, and excellent performance. For example, open-source
distributed object stores that offer these benefits include RIAK \cite{Riak}, Swift \cite{swift}, HDFS \cite{borthakur2008hdfs}, CephFS \cite{tahoe} and Quantcast QFS \cite{qfs}.
The main idea behind these systems is to replicate and spread the data uniformly
across a cluster of servers, in order to increase availability and reliability and take advantage of data declustering
(in other words, process data in parallel across multiple servers).

These systems make good use of the available cluster resources, but they are targeted towards specific query workloads.
In particular, replication and declustering favor MapReduce-style processing, or any processing that
can be easily parallelized. On the other hand, declustering hurts the performance of query workloads
that perform certain types of joins, the reason being that joins could result in a large volume of data migrations, which can
saturate the network and reduce the performance of the underlying distributed data store.  For these types of queries,
careful placement of data around the cluster is critical for guaranteeing high availability, reliability and high query throughput.

For example, in a relational data warehouse, queries include multi-way joins to combine fact tables with multiple dimension tables.  Materialized views are commonly used to pre-compute query results and must be maintained over time.  Thus, data placement in a distributed data warehouse is critical for good query performance and efficient view maintenance; the latter is particularly important in on-line streaming warehouses \cite{golab09stream}, which are continuously updated (as opposed to being taken down for refresh once a week or once a month) and must keep up with the incoming data feeds.

Recent work such as CoHadoop \cite{eltabakh11cohadoop} enables the co-location of related data, but places the burden on the user or the database administrator to determine what should be co-located and where.  In this paper, we aim to automate this process by proposing algorithms for computing nearly-optimal \emph{data placement} strategies for a given workload.

In the simplest version of the problem, we are given a set of base tables, ad-hoc queries and servers with known capacities.  The goal is to decide where to store the tables and where to evaluate the queries in order to minimize the data communication cost during query evaluation.
We also address other issues that arise in our motivating applications, such as load balancing, replication, materialized views, and complex query plans consisting of sequences of intermediate results that may be computed on different servers.

\subsection{Contributions and Roadmap}

First, not surprisingly, we show that
even the simplest formulation of our data placement problem 
is NP-Hard.

Second, and more surprisingly, we reduce the data placement problem without replication to graph partitioning.  
Previous work on distributed data placement relied on \emph{hypergraph} partitioning to represents different objectives 
such as minimizing the number of distributed transactions \cite{curino10schism,kayyoor13data}.  
As we will show, hypergraph partitioning fails to capture the data communication cost, but our reduction to graph partitioning can do so \emph{exactly}.  We consider this reduction to be the main contribution of this paper.

The reduction to graph partitioning is a desirable result.  Graph partitioning is an NP-Hard problem, but it has been studied extensively.  In particular, effective and efficient approximation algorithms and a variety of robust libraries are publicly available (e.g., METIS \cite{karypis95metis} and Chaco \cite{hendrickson94chaco}).  Some problems involving hypergraph partitioning may be approximated by graph partitioning, as was done in  \cite{curino10schism,kayyoor13data} (see Section~\ref{sec.related}).  On the other hand, by reducing our problem to graph partitioning, we can directly use tools such as METIS to provide good approximate data placement plans very efficiently in practice.
Had we needed hypergraph partitioning, we would have
had to use {\em hypergraph} partitioning software, which, by solving
a more general problem, gives worse performance.

Third, we present two classes of algorithms for our problem: non-trivial integer program (IP) formulations that compute an optimal data placement using an IP solver, and practical approximation algorithms based on our reduction to graph partitioning.  We also extend our algorithms to handle load balancing.

Fourth, we turn our attention to the version of the problem with replication
and we present two heuristics that use the reduction to graph partitioning as a subroutine.
We also present an IP formulation that can be used to compute an optimal solution.

Fifth, we present a detailed empirical evaluation of our algorithms using the TPC-DS decision support benchmark.  For many parameter settings, it may take an IP solver weeks (or
longer) to compute an optimal data placement, whereas our reduction produces nearly-optimal solutions in seconds.

The remainder of the paper is organized as follows.
Section \ref{sec.definition} formally defines the problem and notation used throughout the paper.
Section \ref{sec.related} presents related work.
Section \ref{sec.solution} presents our reduction to graph partitioning.
Section \ref{sec.experiments} presents an empirical study.
Finally, Section \ref{sec.conclusion} concludes the paper.

\section {Problem Definition}
\label{sec.definition}
To simplify the presentation, first we define a simple version of the problem with queries and base tables.
We will address extended versions, including load balancing and replication, in Section~\ref{sec.solution}.
\begin{definition}[Data Placement]
\label{def.problem}
Given
\begin{enumerate}
\item $n$ tables $T=\{T_1, T_2, \ldots, T_n\}$, the $j$th table having a nonnegative integral size $t_j$,
\vspace{-5pt}
\item $m$ queries $Q=\{Q_1, Q_2, \ldots, Q_m\}$, each referencing one or more tables, where $Q_i \subseteq T$,
\vspace{-5pt}
\item $l$ servers $S=\{S_1, S_2, \ldots,S_l\}$, the $k$th server having nonnegative integral storage capacity $s_k$, and
\vspace{-5pt}
\item for each $Q_i$ and $T_j \in Q_i$, a nonnegative integral communication cost $C_i^j$, which is the cost
incurred for transferring whichever part of table $T_j$ is needed in order to evaluate query $Q_i$ (e.g., after performing any local projections or selections),
\end{enumerate}
assign each table in $T$ to one of the $l$ servers in $S$ so as not to violate any capacities,
while minimizing the communication cost paid to process all queries.
\end{definition}

Let a {\em placement} of tables to servers be a mapping $f: T \rightarrow S$. A placement is {\em legal} if
$\sum_{j:f(T_j)=S_k}{t_j \le s_k}$ for $1 \le k \le l$, i.e., the total size of all tables placed on server $S_k$
is no greater than $s_k$. 

The cost of a legal placement is defined as follows. We allow any query to be processed
on any server. The cost of processing query $Q_i$ on server $S_k$ is the communication cost
of shipping to $S_k$ all those (required fragments of) tables in $Q_i$ which are not stored on $S_k$, i.e.,
$\sum_{T_j \in Q_i: f(T_j) \neq S_k}{C_i^j}.$ We call this query processing model \emph{query-site execution}.  Since our goal is to minimize the overall communication cost, we will
assume from now on that we choose the server on which to process $Q_i$ so as to minimize the associated
communication cost:
\begin{definition}[Query Communication Cost]
\label{def.cost}
Given a placement $f$, the communication cost of processing query $Q_i$ is
$$cost(f, Q_i)=\min_k{\left[ \sum_{j:T_j \in Q_i, f(T_j) \neq S_k}{C_i^j} \right]}.$$
\end{definition}

For example, suppose query $Q_1$ joins three tables, $T_1$, $T_2$ and $T_3$.  For simplicity, assume that $C_1^1=C_1^2=C_1^3$, i.e., $Q_1$ requires equally-sized fragments of each of the three tables.  If $T_1$ and $T_2$ are placed on $S_1$ and $T_3$ is placed on server $S_2$, then we will evaluate $Q_1$ on server $S_1$ rather than $S_2$.  The former has a communication cost of $C_1^3$ (we need to ship this fragment of $T_3$ to $S_1$) whereas the latter has a communication cost of $C_1^1 + C_1^2$ (we need to ship the $T_1$ and $T_2$ fragments to $S_2$).

\textbf{Note:} While our problem definition assumes that the objects to be stored are tables, we can easily extend it to the case of partitioned tables.  In this case, each ``table'' is actually a part of a larger table, and, for instance, a distributed hash join can be modelled as a set of smaller queries that join the corresponding parts.  How to partition the tables for a given workload is an orthogonal problem that has been addressed, e.g., in \cite{NehmeB11}.

Given that the table sizes in our problem definition are arbitrary nonnegative integers, determining
if there is {\em any} legal placement  at all is NP-Hard, since {\sc Data Placement} is at least as hard
as {\sc Partition}, even when $k=2$. (The proof is straightforward and appears in the appendix.)
Since the feasibility question is itself NP-Hard,
there can be no polynomial-time approximation algorithm with any fixed ratio (unless
P$=$NP). This means that in order to get good algorithms, assuming that P$\neq$NP, we
will have to allow, in the worst case, some overloading of the servers.

\section{Related Work}
\label{sec.related}
Given a set of tables, a set of servers and a set of queries, Kayyoor et al.\ \cite{kayyoor13data}
studied algorithms for deciding which tables to replicate and where to place those replicas.  Their objective is to minimize the average query span, 
which is the number of servers involved in answering a query. Minimizing the average query
span of a query workload is different from minimizing the communication cost for the same workload, as illustrated by the following example. Assume we have three servers and one query associated with six tables, as shown in Figure \ref{fig.kayyoor13data}.  As illustrated, tables $T_1$ through $T_4$ are larger than $T_5$ and $T_6$.  Assume for simplicity that in order to evaluate the query we need to use all tuples contained in all six tables. Two possible placements are shown in the figure, both of which give a query span of three since the tables needed by the query are spread out on three servers.  However, the communication cost of Placement 2 is more than double that of Placement 1---in both cases, it is best to execute the query on server 1, but Placement 2 additionally requires two large tables, $T_3$ and $T_4$, to be shipped there. 
Since the algorithm of
Kayyoor et al. cannot differentiate between these two cases, it cannot be used to solve our problem.
The authors do present a generalization of the algorithm that considers table sizes, but only to guarantee that the capacities of the servers are not exceeded by the recommended
placement. Table sizes play no role in deciding a placement that minimizes the communication cost of
the query workload.

\begin{figure}
\centering
\includegraphics[width=1.8in]{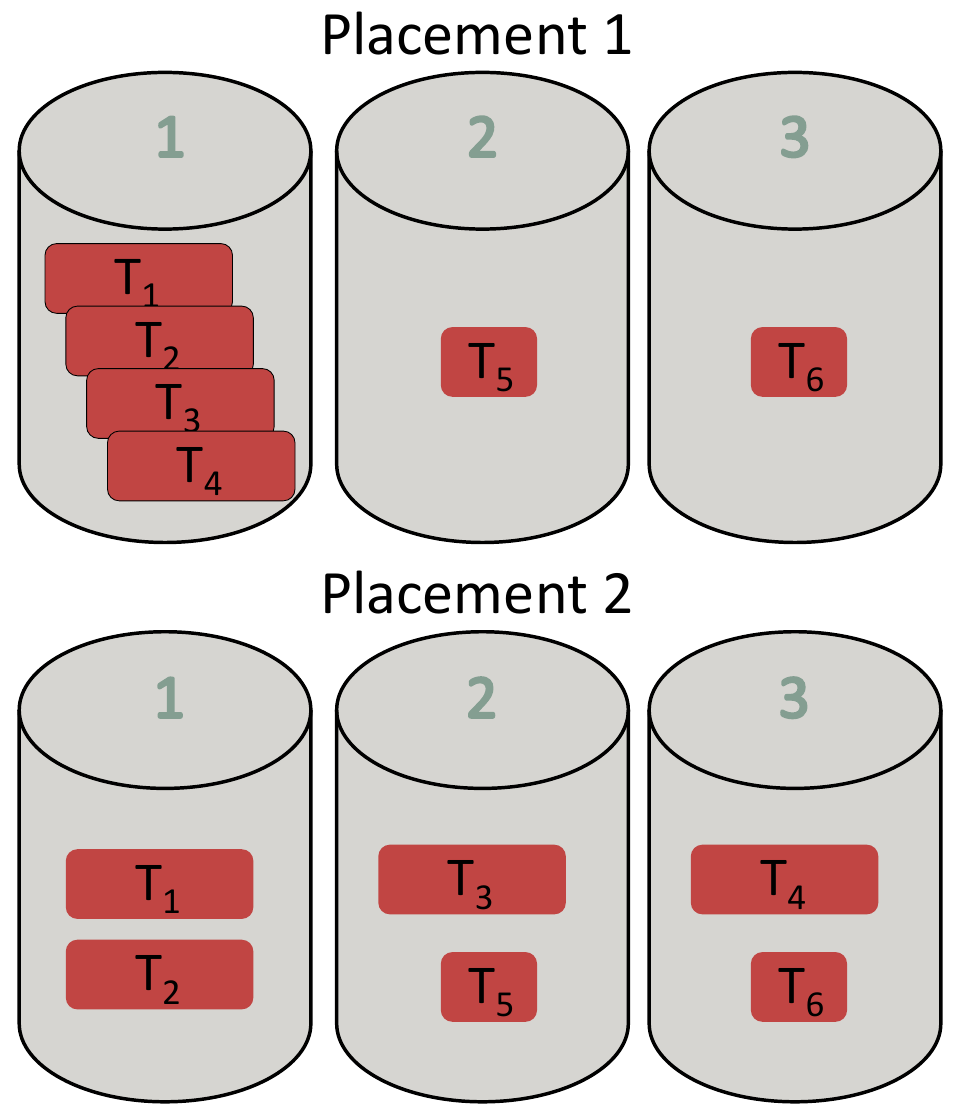}
\caption{Two placement schemes with equal query span but different communication cost, for one query over six tables.}
\label{fig.kayyoor13data}
\end{figure}

Curino et al.\ \cite{curino10schism} present Schism, which
minimizes the number of distributed transactions in a given workload by assigning and replicating
individual tuples to servers.
The database and the workload are represented as a graph, whose nodes correspond to tuples and whose edges connect tuples accessed in the same transaction.
Graph partitioning algorithms are applied to find balanced parts, each part corresponding to one server, that minimize the weight of the cut edges.
As in Kayyoor et al. \cite{kayyoor13data}, 
this technique does not differentiate between cases where the required tuples for
a transaction are distributed across two or more than two servers. Hence, it cannot distinguish
between the two data placement plans shown in Figure \ref{fig.kayyoor13data} and cannot be used to solve our problem.

In addition to database design, graph partitioning was used to solve data placement problems in parallel computing; see, e.g., \cite{HendricksonK00}.  Again, this work optimized for different objectives and cannot be adapted to solve our problem.

A natural way of thinking about the data placement problem, which was considered by both Kayyoor et al.\  \cite{kayyoor13data} and Curino et al.\ \cite{curino10schism},
is to model the workload as a hypergraph, in which tables or tuples are represented as nodes, and queries or transactions are represented as hyperedges.
Each hyperedge is a set of nodes, corresponding to the tables associated with the particular query.\footnote{When queries are associated with at most two tables each, the hypergraph representation reduces to a graph representation.
}
Interestingly, the optimization
objectives of Kayyoor et al.\ and Curino et al.\ are captured exactly by hypergraph partitioning.  For example, cutting a hyperedge means that the tuples needed by this particular transaction are placed on multiple servers; therefore, the number of cut hyperedges exactly corresponds to the number of distributed transactions.
However, since graph partitioning can be solved more accurately than the more general problem of hypergraph partitioning, Kayyoor et al.\ and Curino et al.\ provide reductions from hypergraph to graph partitioning, but these reductions are not exact.
On the other hand, hypergraph partitioning cannot capture data communication costs, because it is impossible
to assign hyperedge weights appropriately and decide how to distribute a hyperedge weight when that
edge is split across multiple servers.  Recall the example from Section~\ref{sec.definition} involving query $Q_1$.
In the hypergraph representation, $Q_1$ induces a hyperedge containing $T_1$, $T_2$ and $T_3$.
Clearly, the weight of this hyperedge, i.e., the communication cost paid by $Q_1$, depends on the placement of the
three tables and of $Q_1$ itself.  But it is the algorithm's job to determine this placement, so it is not possible to assign
an accurate edge weight a priori.
It is easy to show using adversarial counterexamples that, irrespective of how
hyperedge weights are chosen, the data communication cost obtained by hypergraph
partitioning can be arbitrarily worse than the optimal cost. 

It is surprising that while hypergraph partitioning fails to capture the data communication cost, we can provide an exact reduction of this problem to a standard graph partitioning instance. 
As we shall see in Section~\ref{sec.solution}, our graph construction is different from previous constructions which used hypergraphs, {\em even in the special case when the hyperedges all have size two, i.e., when the hypergraphs are graphs.}  Rather than building a graph with tables as nodes and queries as edges, the trick will be to build a \emph{bipartite} graph with queries on one side and tables on the other.

Finally, we note that partitioning algorithms for relational databases have been studied extensively in the past, but the focus has been on data declustering and physical design tuning in order to speed up query evaluation by taking
advantage of parallelization. Partitioning strategies include range partitioning, hash partitioning, and
partitioning based on query cost models
\cite{sanjay04integrating, zilio98physical, rao02automating, navathe89vertical, ghandeharizadeh90hubrid, liu96partitioning, NehmeB11}.
Furthermore, modern distributed storage systems, such as BigTable \cite{chang06bigtable}, are not optimized for
relational workloads on multiple tables. In addition, distributed key/value stores, such as
Amazon Dynamo \cite{dynamo}, HDFS \cite{borthakur2008hdfs}, RIAK \cite{Riak}, and Quantcast QFS \cite{qfs} focus primarily on randomly distributing redundant data (either by replication
or erasure coding)
with the main objective being to increase reliability and availability. Recently, systems like CoHadoop \cite{eltabakh11cohadoop}
have been developed to take advantage of non-random data placement in order to speed up evaluation of
certain classes of queries. Nevertheless, these systems focus only on technical issues and leave the responsibility of
choosing the placement of data to the database administrator. Our work automates this process.

\section{Our Solution}\label{sec.solution}
First, we present an \emph{exact} reduction from the {\sc Data Placement} problem in Definition~1 to {\sc Graph Partitioning} (Section~\ref{sec.reduction}). Then we present an IP formulation that can be used with any IP solver to produce an optimal solution (Section~4.1.1).
Next, we generalize our definition to include arbitrary query execution plans with materialized views and intermediate results, and we present a generalized reduction to {\sc Graph Partitioning} (Section~\ref{sec.materialized_views}).  Section~4.3 discusses load balancing.
Finally, we discuss how to handle replication in Section~4.4. We give an IP formulation for finding an optimal placement
when a fixed number of replicas per table is desired, and we discuss two heuristics that use our reduction
to find good placements efficiently.

\subsection{Reduction To Graph Partitioning}
\label{sec.reduction}

Our version of {\sc Graph Partitioning} is  defined as follows. Given a node- and edge-weighted
graph $G=(V,E)$, and a sequence of $l$ nonnegative capacities, partition $V$
into $l$ parts such that the total weight of the nodes in
the $k$th part is at most the $k$th capacity, so as to
minimize the sum of the weights of the (cut) edges whose endpoints are in different parts.

Throughout this paper, the term ``partition'' will always refer
to what is known in the literature as an {\em ordered partition},
i.e., a sequence of disjoint subsets whose union is the universe.
Equivalently, an ordered partition of $V$ into $l$ parts is just a
mapping $f:V\rightarrow \{1,2,...,l\}$.

More formally, given a graph
$G=(V,E)$ with edge $e$ having weight $w(e)$, node $v$ having
weight $w(v)$ (the distinction between nodes and edges
keeping the notation unambiguous), and a sequence
$\langle s_1, s_2, \ldots, s_l\rangle$ of $l$ nonnegative integers,
let a mapping $g: V\rightarrow \{S_1,S_2,...,S_l\}$
(which defines an ordered $l$-part partition) be {\em legal} if $\sum_{v: g(v)=S_k} w(v) \le s_k$, for
$1 \le k \le l$; then the goal is to find a legal $g$ so as to minimize
$$\sum_{1\le k<k'\le l} \left [ \sum_{u<v: \{u,v\}\in E, g(u)=S_k, g(v)=S_{k'}}  w(\{u,v\})\right ].$$

The construction of the graph used for {\sc Graph Partitioning}
to solve  {\sc Data Placement} is as follows. We construct a
\emph{bipartite} graph in which
there is a node for each query on the left side and a node for each table on the right side of the graph. There is an edge
between a query $Q_i$ and a table $T_j$ iff $T_j \in Q_i$. The weight of an edge is equal to $C_i^j$. The main benefit of this
construction comes from the fact that we use, in addition to
nodes for tables,  a separate node for each query, unlike in Kayyoor et al.\ and Curino et al.\ where
there are nodes {\em only} for tables/tuples.

The crucial benefits of our construction vis-a-vis previous papers
on data placement
are that it preserves the
optimal cost {\em exactly} and that it generates {\em graphs},
not hypergraphs.  Having hyperedges of arbitrarily large size,
hypergraphs are far more general than regular graphs (whose
edges all have size 2), making hypergraph partitioning a far more general, and therefore much
more difficult,
problem to approximately solve than graph partitioning.

An example of our construction is shown in Figure \ref{fig.bipartite}, in which the query workload consists of four queries,
involving six tables. The queries are $Q_1=\{T_1, T_4, T_5\}$, $Q_2=\{T_1, T_3, T_6\}$, $Q_3=\{T_1, T_5\}$, $Q_4=\{T_4\}$.
Tables $T_1, T_2, T_3, T_4$ have size $2$. Tables $T_5, T_6$ have size $1$, and each
server has capacity $4$.  For simplicity, in some subsequent examples we assume that $C_i^j=t_j$ for all $i,j$; however, in practice, it is usually the case that $C_i^j<t_j$ since data-reducing operations such as projections and selections can easily be done locally.

\begin{figure}
\centering
\includegraphics[width=2.5in]{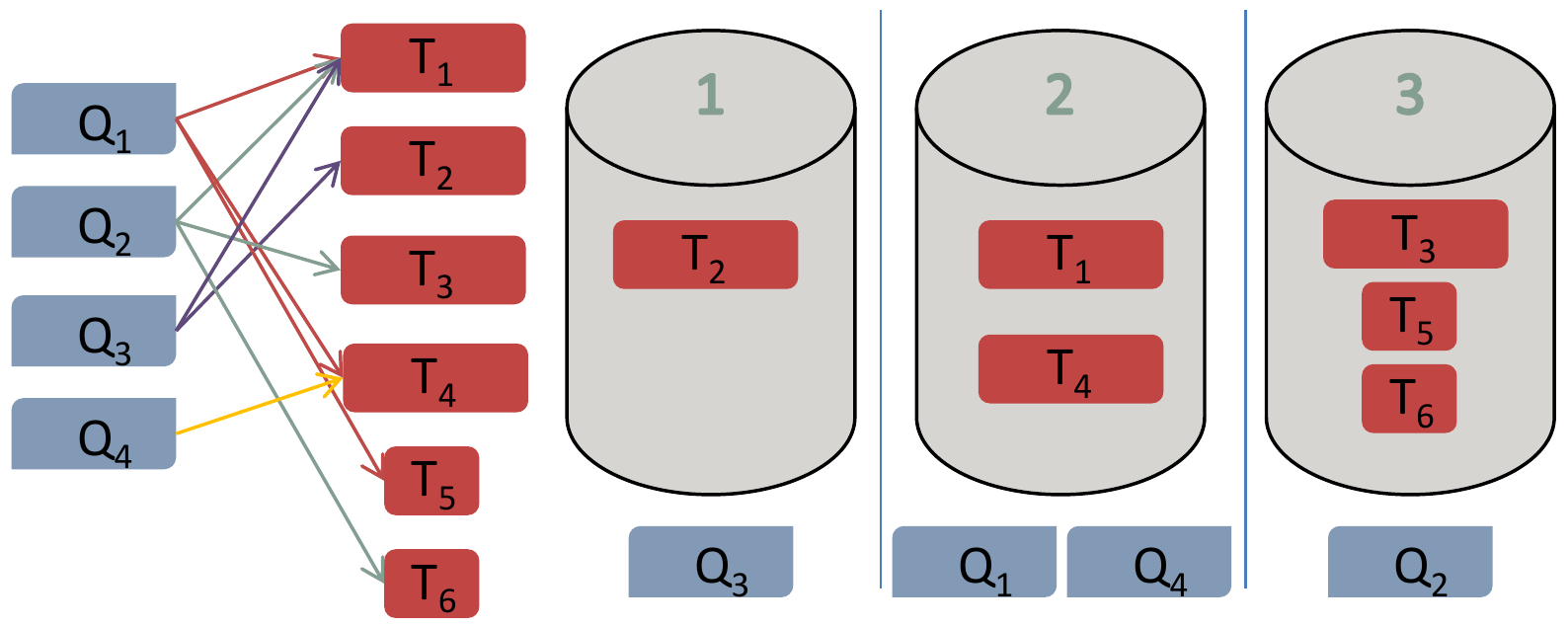}
\caption{A query workload and the corresponding bipartite graph. The partition shown corresponds to the optimal
placement of tables on three servers so as to minimize the data communication cost.
For simplicity, we assume that $C_i^j = t_j$ for all $i, j$. 
With the given placement, $Q_1$ will be executed on server $2$, $Q_2$
on server $3$, $Q_3$ on server $1$ and $Q_4$ on server $2$.}
\label{fig.bipartite}
\end{figure}

\begin{theorem}
\label{thm.data_placement}
There is a (simple) polynomial-time  transformation that takes an
instance $I$ of {\sc Data Placement} and produces an instance
$I'$ of {\sc Graph Partitioning}
such that (1) given any feasible solution to $I$, there is a
feasible solution to $I'$ of no greater cost, and
(2) given any feasible solution to $I'$, there is a feasible
solution to $I$ of no greater cost.
Furthermore, there are (trivial) polynomial-time algorithms
which convert between the specified solutions for instances $I$ and $I'$.
\end{theorem}

\begin{proof}
Take the given instance $I$ of {\sc Data Placement} and build an
instance $I'$ of {\sc Graph Partitioning} as follows.  The graph
$G=(V,E)$ where $V=Q \dot \cup T$.  There is an edge of $E$ from $Q_i$ to
$T_j$ iff $T_j\in Q_i$;  the weight of this edge is $C_i^j$.
The weight $w(T_j)$ of $T_j \in T$ is $t_j$;  the weight $w(Q_i)$ of $Q_i \in Q$ is 0.
The $k$th capacity is $s_k$.  This completes the description of $I'$.

First, given a feasible solution $f$ for $I$, we show how to construct a feasible solution $g$ for $I'$ of no greater cost.
Given a feasible solution $f$ for $I$, we just define
$g(T_j)=f(T_j)$ for all $T_j\in T$ and
define $g(Q_i)=S_m$ where
$$m = \arg\,\max_k \sum_{j: T_j \in Q_i, f(T_j)=S_k} C_i^j.$$
We want to show first that $g$ defines a feasible solution for $I'$
whose cost is at most the cost of $f$ on $I$.

We know that for all $k$, $$\sum_{j: f(T_j)=S_k} t_j\le s_k.$$
Since $w(Q_i)=0$ for all $i$, for all $k$ we have
$$\sum_{v \in V: g(v)=S_k} w(v)= \sum_{j: g(T_j)=S_k} w(T_j)=\sum_{j: f(T_j)=S_k} t_j\le s_k;$$
therefore the partition is legal and hence $g$ defines a feasible
solution for $I'$.  The cost of the partition is
$$\sum_{k<k'} \left [\sum_{u<v : \{u,v\}\in E, g(u)=S_k, g(v)=S_{k'}} w(\{u,v\})\right ]$$
$$=\sum_{i: Q_i\in Q} \left [\sum_{j: T_j\in Q_i, g(T_j) \ne g(Q_i)} C_i^j \right ]$$
$$=\sum_{i: Q_i\in Q} \left [\sum_{j: T_j\in Q_i, f(T_j)\ne g(Q_i)} C_i^j\right ]$$
$$=\sum_{i: Q_i\in Q}\left[\sum_{j: T_j\in Q_i} C_i^j - \sum_{j: T_j\in Q_i, f(T_j)=g(Q_i)} C_i^j\right ]$$
$$=\sum_{i: Q_i \in Q} \left[ \sum_{j: T_j \in Q_i} C_i^j - \max_k \sum_{j: T_j\in Q_i, f(T_j)=S_k} C_i^j\right ]$$
$$=\sum_{i: Q_i \in Q} \min_k \left[ \sum_{j: T_j \in Q_i, f(T_j) \neq S_k} C_i^j \right]$$
$$=\sum_{i: Q_i \in Q} cost(f,Q_i)=\mbox{cost of }f\mbox{ on }I.$$

Second, given a feasible solution $g$ for $I'$, we show how to
construct a feasible solution $f$ for $I$ of no greater cost.
Given any feasible solution $g$ for $I'$, we may assume that for all $i$, $g(Q_i)= S_k$
where $k$ is chosen to maximize
$$\sum_{j: T_j\in Q_i,g(T_j)=S_k} C_i^j,$$
for modifying $g$ in order to attain this property for all $i$
cannot increase the cost.  Now we just define $f(T_j)=g(T_j)$
for all $j$.  We must show that $f$ defines a feasible solution
to $I$ whose cost is at most the cost of $g$ on $I'$.

Because $g$ is feasible for $I'$, we have, for all $k$,
$$s_k\ge \sum_{j: g(T_j)=S_k} w(T_j)=\sum_{j: g(T_j)=S_k} t_j=\sum_{j: f(T_j)=S_k} t_j;$$
therefore the placement for $I$ is legal.
The cost of solution $f$ for $I$ is
$$\sum_{i: Q_i\in Q} \left [ \sum_{j: T_j\in Q_i} C_i^j - \max_k\left [ \sum_{j: T_j\in Q_i, f(T_j)=S_k} C_i^j \right ]\right ]$$
$$=\sum_{i: Q_i\in Q} \left [ \sum_{j: T_j\in Q_i} C_i^j- \max_k\left [ \sum_{j: T_j\in Q_i, g(T_j)=S_k} C_i^j \right ]\right ]$$
$$=\mbox{cost of }g\mbox{ on }I'. \qed$$
\end{proof}

It is easy to extend the construction of the bipartite graph to handle query frequencies while computing the communication cost. Let $\nu_i \in \mathbb{N}$ denote the frequency of query $Q_i$. Given a placement $f$, the communication cost of processing $Q_i$  with frequency $\nu_i$ is
$$cost(f, Q_i)=\min_k{\left[ \sum_{j:T_j \in Q_i, f(T_j) \neq S_k}{\nu_i C_i^j} \right]}.$$
To change the construction, we need only to update the weight of the edge from query $Q_i$ to table $T_j$ from $C_i^j$ to $\nu_i C_i^j$.

\subsubsection{An Integer Program}

Building an IP for {\sc Data Placement} can be done as follows.
The IP has a boolean variable $x_{Z,S_k}$ if $Z$ is either a query
or a table and $S_k$ is a server, the meaning of which is that
$x_{Z,S_k}=1$ if query or table $Z$ is stored on server $S_k$ and 0
otherwise. For each query or table $Z$, we add the constraint
$\sum_{k}{x_{Z,S_k}} = 1$,
meaning that every query or table is assigned to exactly one server.
For each server $S_k$, we have a constraint preventing $S_k$
from being overloaded, namely, $\sum_{\mbox{query or table }Z}
z_Z\cdot x_{Z,S_k}\le s_k$,
where $z_Z$ is (a) $0$ if $Z$ is a query or (b) equal to the size of the table if $Z$ is a table. Here $s_k$ is the capacity of $S_k$.
This completes the description of constraints.

As for the objective function, it is the sum, over $i,j$ such that $T_i\in Q_j$,
of $C_i^j$ times [$1$ if $Q_i$ and $T_j$ are stored on
different servers and 0 otherwise]. More formally:
\[
\min \sum_{i:Q_i \in Q}\sum_{j: T_j \in Q_i} C_i^j \lambda_i^j,
\]
in which $\lambda_i^j = \max_{k} |x_{Q_i,S_k} - x_{T_j,S_k}|$,
which equals 0 if and only if both $Q_i$ and $T_j$ are stored on server $S_k$ and 1 otherwise, as desired.

(In reality we cannot define $\lambda_i^j$ as stated since
absolute value is nonlinear and integer linear programs must be
linear.  Instead we just require that
$\lambda_i^j$ be greater than or equal to both
$x_{Q_i,S_k}-x_{T_j,S_k}$ and
$x_{T_j,S_k}- x_{Q_i,S_k}$
for all $S_k\in S$. Then, in an
optimal solution, since the problem is a minimization,
$\lambda_i^j$ will equal $\max_{k} \max\{
x_{Q_i,S_k}-x_{T_j,S_k}, x_{T_j,S_k}-x_{Q_i,S_k}
\}=\max_{k}
|x_{Q_i,S_k}-x_{T_j,S_k}|$.)

\subsection{Arbitrary Query Execution Plans}
\label{sec.materialized_views}

Definition 1 explicitly distinguishes between base tables
and queries.  In practice, frequently-used queries are materialized to improve performance;
hence some queries can also act as tables.
Furthermore, complex queries (e.g., ones that contain several
join, aggregation, selection, and projection predicates) can be
decomposed into several steps.  
This gives rise to complex distributed query plans, which may be represented by 
directed acyclic graphs, in which the intermediate results of one step are pipelined to the next step 
(or materialized in a \emph{temporary table} and fed
to the subsequent step).  In a distributed setting,
each node of the DAG can be evaluated on a different server, where intermediate results
need to be shipped to the server(s) responsible for executing subsequent steps, but need not
be stored permanently.
Query execution plans can be determined by a query optimizer
using cost-based and rule-based optimization strategies, so as to minimize the cost of evaluating
a query as well as minimizing the amount of intermediate data produced \cite{chaudhuri98overview}.
Given arbitrary query plans,
we want to find the best data placement strategy in order to minimize the total data communication cost
required to evaluate all such plans.

Fortunately, generalizing our reduction
is simple.
The trick is to view tables, queries, intermediate nodes of query execution DAGs and
materialized views {\em all} as views.
In the new problem, any view $V_i$ can depend on any other $V_j$.
(Earlier, queries  depended only on tables,
tables depended on nothing, and nothing depended on queries.)
In fact, we will allow each view to be computed on one server
and stored possibly on a different one (however, it is easy to force each view to be computed and stored on one server if so desired).
We assume that all servers have enough
overflow capacity to temporarily store all intermediate results and query results needed.

We define the {\sc Generalized Data Placement} problem as follows.

\begin{definition} ({\sc Generalized Data Placement}) (or {\sc
GDP}).
Given
\begin{enumerate}
\item
$n$ views $\{V_1,V_2, ...,V_n\}$, the $j$th having a nonnegative
integral size $t_j$,
\vspace{-5pt}
\item
a set $S=\{S_1,S_2,..., S_l\}$ of $l$
storage servers, the $k$th having
nonnegative integral storage capacity $s_k$,
\vspace{-5pt}
\item
a nonnegative integral {\em transfer cost} $m_j$ for each view $V_j$,
\vspace{-5pt}
\item
a directed acyclic graph $K=(V,E)$ representing the union of execution plans of all views
(an arc $(V_i,V_j)$ meaning that view $V_i$ needs view $V_j$),
\vspace{-5pt}
\item
for each $i,j$ such that $(V_i,V_j)\in E$, a nonnegative integral communication
cost $C_i^j$ (which is the cost incurred for transferring
whatever part of view $V_j$ is needed in order to evaluate view
$V_i$),
\end{enumerate}
find a {\em computation server} $cs(V_j)\in \{S_1,S_2,...,S_l\}$
and a {\em storage server} $ss(V_j)\in \{S_1,S_2,...,S_l\}$ for
each $V_j$, to minimize the total communication cost, which is
defined to be $$\sum_{i,j:(V_i,V_j)\in E, ss(V_j)\ne cs(V_i)}
C_i^j+\sum_{i: cs(V_i)\ne ss(V_i)} m_i,$$
while satisfying the property that for all $k$,
$$\sum_{i: ss(V_i)=S_k} t_i\le s_k.$$
\end{definition}
The communication cost includes
processing each $V_i$ on its computation server and
transferring each $V_i$ from its computation server to its storage
server, if necessary.
Note that we are assuming, as mentioned above, that $V_i$ can be computed on its
computation server in scratch space on that server, since we are not
including the size of $V_i$  in the space needed on its computation
server.

Now we show that a simple generalization of the graph construction for the case
of tables and queries works for {\sc GDP}.
Specifically, build a node- and edge-weighted bipartite graph
$H=(V',V,F)$ where $V'=\{V'_i|V_i \in V\}$ is a copy of $V$
and $F=\{\{V_j,V'_i\}|(V_i,V_j)\in E\}\cup
\{\{V_i, V'_i\}|i=1,2,...,n\}$.
The weight of an edge
$\{V_j,V'_i\}$ with $j\ne i$ is $C_i^j$;  the weight of an edge
$\{V_i,V'_i\}$ is $m_i$.   The weight of node $V'_j$ is $0$;
the weight of $V_j$ is $t_j$.
We seek a minimum-edge-weight partition of $H$
into $k$ parts such that the total node weight of the $k$th part
is at most $s_k$.

If $V_j$ is a query, set $t_j=0$ and
$m_j=\infty$ to represent the fact that queries are not stored
(effectively, they can be stored for free anywhere).
If $V_j$ is a base table, set $m_j=\infty$ to force the graph partitioning
algorithm to not cut the edge $(V_j,V_j')$ and hence to have $ss(V_j)=cs(V_j)$. If $V_j$ is a materialized view (i.e.,
a node with both in-arcs and out-arcs which is to be stored),
set $m_j=t_j$ to represent the size of the computed view. 
If $V_j$ is an intermediate result (i.e., a node with both in-arcs and
out-arcs which is to be computed but not stored permanently),
set $t_j=0$ but let $m_j$ equal the size of the computed intermediate query result.

This construction is a generalization of the previous graph
construction for for {\sc Data Placement} in the sense that when every view is a table (i.e.,
depends on nothing) or a query (i.e., upon which nothing
depends), in $H$, nodes $Q_i$ on the left side of the bipartite
graph and nodes $T'_j$ on the right side are adjacent (only) to $Q'_i$ and $T_j$, respectively, by
infinite-weight edges, and hence in any finite-edge-cost
partition will be combined together.  Once those
nodes are combined, $Q_i$ with $Q'_i$ and $T_j$ with $T'_j$, what remains
is exactly the weighted graph constructed earlier.

\subsubsection{An Example}
Here we give a concrete example to show how one would use the
reduction specified above to solve an instance of {\sc GDP}.
Suppose there are seven views $V_1,...,V_7$, and that the
DAG $K$ is as shown in Figure \ref{fig.GDP_example}.
Nodes $V_1,V_2,V_3$ are base tables;  node $V_7$ is a query;
node $V_4$ is an intermediate result;
and nodes $V_5$ and $V_6$ are materialized views.
Because $V_7$ and $V_4$ will not be stored, we set
$t_7=t_4=0$.   For the other nodes, we use the sizes of the
views, e.g., $t_1=8,t_2=5, t_3=4, t_5=10,t_6=7$.

Now we discuss the $m_j$'s.  Since a table should not be
moved, we set $m_j=\infty$ for $V_1,V_2,V_3$.  Because an intermediate result is not stored or moved,
we set $m_j=\infty$ also for $V_4$.  The remaining $m_j$'s are the costs
of transferring the views, hence $m_j = t_j$.

For simplicity, let us set $C_i^j$ to be the size of view $V_j$.    In other words, whenever
view $V_i$ depends on view $V_j$, we must transfer all of $V_j$
from the storage server of $V_j$ to the computation server of $V_i$.
Unless $t_j=0$, which signifies that $V_j$ is not to be stored,
$C_i^j $ will just equal $t_j$.  This is the case for
$j=1,2,3,5,6$.
In our example, $t_4=t_7=0$.
Let us take $C_i^4=8$ for all $i$.  We do not need $C_i^7$.

Suppose there are $k=2$  servers $S_1,S_2$ with capacities $s_1=s_2=18$.
We now build a bipartite graph with left-hand-side nodes $V_1,V_2,...,V_7$ and
right-hand-side nodes $V'_1, V'_2, ...,V'_7$ and the following edges:
\begin{itemize}
\item
$\{V_1,V'_1\}$,
$\{V_2,V'_2\}$,
$\{V_3,V'_3\}$, and
$\{V_4,V'_4\}$, all of weight $\infty$, $\{V_5,V'_5\}$ of
weight 10, $\{V_6,V'_6\}$ of weight 7, and $\{V_7,V'_7\}$ of
weight 0;
\item
$\{V_1,V'_4\}$ of weight 8, $\{V_2,V'_4\}$ of weight 5,
$\{V_1,V'_5\}$ of weight 8,
$\{V_3,V'_5\}$ of weight 4,
$\{V_2,V'_6\}$ of weight 5,
$\{V_3,V'_6\}$ of weight 4,
$\{V_4,V'_7\}$ of weight 8, $\{V_5,V'_7\}$ of weight 10, and
$\{V_6,V'_7\}$ of weight 7.
\end{itemize}
A left-hand-side node $V_j$ has weight $t_j$;  a right-hand-side node $V'_j$ has weight 0.
A picture of this bipartite graph appears in Figure \ref{fig.GDP_example_2}.

\begin{figure}[t]
\centering
\includegraphics[width=2.5in]{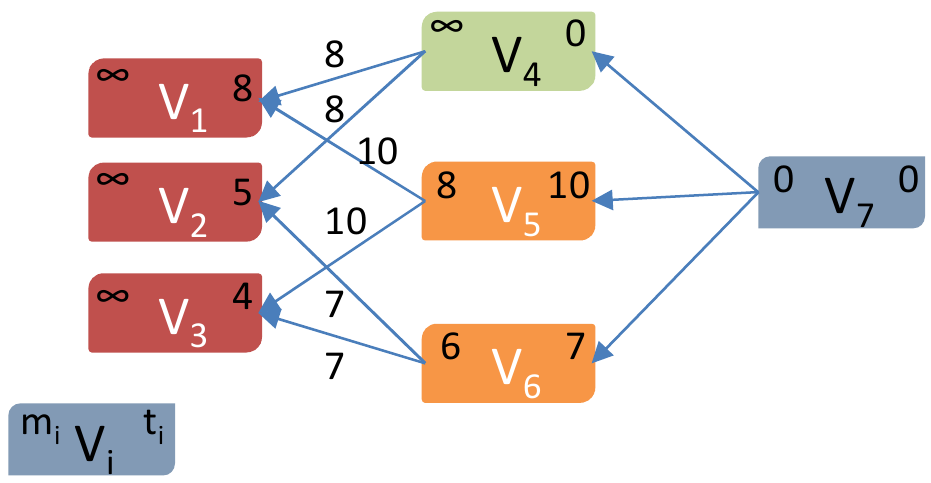}
\caption{Query execution DAG.}
\label{fig.GDP_example}
\end{figure}

\begin{figure}[t]
\centering
\includegraphics[width=1.8in]{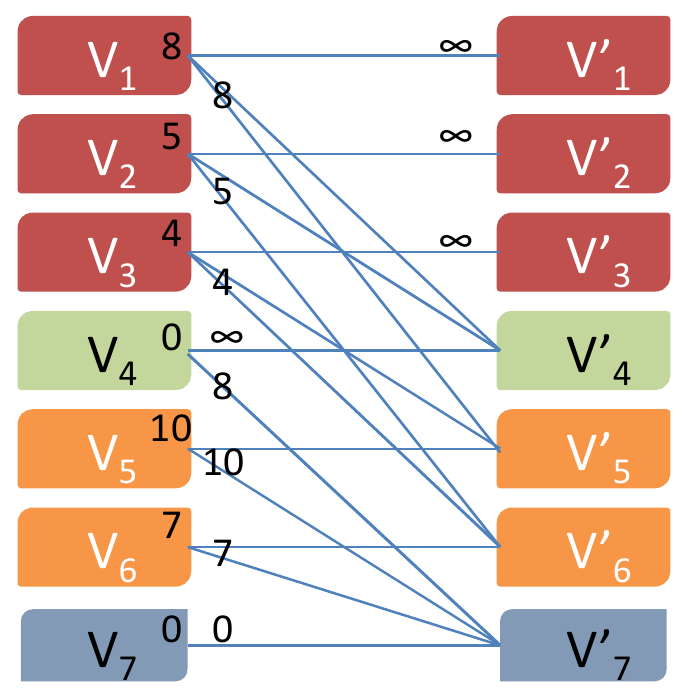}
\caption{Bipartite graph corresponding to Figure \ref{fig.GDP_example}.}
\label{fig.GDP_example_2}
\end{figure}

The goal is to find a partition of the vertex set into $k=2$
parts each of total node weight at most 18, so as to minimize
the cost of the cut edges.  For example, one feasible solution
is to put $V_2,V_3,V_6$ and $V'_2, V'_3, V'_5$ on server $S_1$ and
the remaining nodes ($V_1,V_4,V_5,V_7,V'_1,V'_4,V'_6,V'_7$) on server $S_2$. In other words,
$V_2,V_3$ are computed and stored on server $S_1$;
$V_1,V_4,V_7$ are computed and stored on server $S_2$;  $V_5$ is
computed on server $S_1$  and stored on server $S_2$; and $V_6$ is
computed on server $S_2$ and stored on server $S_1$.
The edges cut by this partition are:
\begin{itemize}
\item
$\{V_1,V'_5\}$, of cost 8, representing the cost of moving
$V_1$ from its storage server $S_2$ to $V_5$'s computation server
$S_1$;
\item
$\{V_2,V'_4\}$, of cost 5, representing the cost of moving $V_2$
from its storage server $S_1$ to $V_4$'s computation server $S_2$;
\item
$\{V_2,V'_6\}$, of cost 5, representing the cost of moving $V_2$
from its storage server $S_1$ to $V_6$'s computation server $S_2$;
\item
$\{V_3,V'_6\}$, of cost 4, representing the cost of moving $V_3$
from its storage server $S_1$ to $V_6$'s computation server $S_2$;
\item $\{V_5,V'_5\}$, of cost 10,
representing the cost of 8 to move $V_5$ from its computation
server $S_1$ to its storage server $S_2$;
\item
$\{V_6,V'_6\}$, of cost 7, representing the cost of 6 to move
$V_6$ from its computation server $S_2$ to its storage server $S_1$;
and
\item
$\{V_6,V'_7\}$, of cost 7, representing the cost of moving $V_6$
from its storage server $S_1$ to $V_7$'s computation server $S_2$.
\end{itemize}

\subsubsection{Proof of Correctness}

Now we state and prove the key property of the new graph
construction.  Within the theorem and its proof, a  ``partition'' refers to an
(ordered) partition into $l$ parts whose $k$th part has node weight
at most $s_k$.
Below the theorem and its proof, we give an example of how the
theorem would be used.

\begin{theorem}
\label{theorem:gdp}
(1) Given any feasible solution
to {\sc GDP}, there is a partition of $H$
of edge cost no greater than the cost of the feasible solution,
and (2)
given any partition of $H$, there is a feasible solution to {\sc
GDP} of cost at most the edge cost of the
partition.
In particular, the optimal value of {\sc
GDP} equals the minimum edge cost of a
partition of $H$.
\end{theorem}

\proof
(1) Take any feasible solution to {\sc GDP}.
For all $k$, we have $$\sum_{i: ss(V_i)=S_k} t_i
\le s_k;$$
the cost of this solution is $$\sum_{i,j:(V_i,V_j)\in E, ss(V_j)\ne
cs(V_i)} C_i^j+\sum_{i: cs(V_i)\ne ss(V_i)} m_i.$$

Define the natural partition $(P_1,P_2,...,P_l)$ of $V\cup V'=V(H)$ by
$P_k=\{V_j|ss(V_j)=S_k\}\cup \{V'_i|cs(V_i)=S_k\}$.
The node weight of $P_k$ is $\sum_{i:V_i\in P_k} t_i=
\sum_{i: ss(V_i)=S_k} t_i$, which we know is at most $s_k$.

Let $part(v)$, for $v\in V\cup V'$, denote $P_k$ such that $v\in
P_k$.
The edge cost of the graph partition is
$$\sum_{i,j: (V_i,V_j)\in E, part(V_j)\ne part(V'_i)}
C_i^j+ \sum_{i: part(V'_i)\ne part(V_i)} m_i$$
$$=\sum_{
i,j:(V_i,V_j)\in E, ss(V_j)\ne cs(V_i)} C_i^j+\sum_{i:
cs(V_i)\ne ss(V_i)} m_i,$$
which is exactly the cost of the {\sc
GDP} solution, so the proof of (1) is complete.

For (2), take any partition $(P_1,P_2,...,P_l)$ of $H$.
We have, for all $k$, $\sum_{j: V_j\in P_k} t_j\le s_k$.  Now
define a solution to {\sc GDP}:
for all $j$, let $ss(V_j)=S_k$, where $V_j\in P_k$, and let
$cs(V_j)=S_{k'}$, where $V'_j\in P_{k'}$.
We will prove that this solution is feasible  for {\sc GDP} and
that its cost equals the edge cost of the graph partition.

We have
$\sum_{j: ss(V_j)=S_k} t_j\le s_k$.
Furthermore, the cost of the defined solution to {\sc GDP}
is $$\sum_{i,j: (V_i,V_j)\in E, ss(V_j)\ne cs(V_i)} C_i^j
+\sum_{i: cs(V_i)\ne ss(V_i)} m_i,$$
which equals
$$\sum_{i,j: (V_i,V_j)\in E, part(V_j)\ne part(V'_i)} C_i^j,
+\sum_{i: part(V'_i)\ne part(V_i)} m_i.$$
This is exactly the edge cost of the graph partition. \qed

Note that for any $V_i$, we can ensure that $ss(V_i)=cs(V'_i)$ by setting $m_i=\infty$.
This ensures that in any finite-edge-cost partition, the edge $\{V_i,V'_i\}$ is not cut, which implies that $ss(V_i)=cs(V'_i)$. Therefore,
if required, materialized views can be required to be computed and stored on the same server.

The integer program for GDP is similar to that for {\sc Data Placement} and is omitted due to space constraints.

\subsection{Load Balancing}
The {\sc Graph Partitioning} algorithm places data optimally across servers but does not
guarantee that queries are also evenly distributed across servers. It is plausible for an
optimal solution to assign the execution of a large percentage of queries to only a small
fraction of servers. Fortunately, there is a straightforward way of distributing
load across servers. 

The idea is to assign to each node of the bipartite graph a
2-dimensional weight $(w_i, l_i)$ and to assign each server a
2-dimensional capacity.
(METIS allows 2-dimensional weights and 2-dimensional capacities, which must
be satisfied componentwise). 
As before, $w_i=t_i$ is the size of the node.
Furthermore, $l_i$ 
corresponds to the estimated execution cost
for all nodes
on the left-hand side of the graph that corresponds to queries, materialized views, temporary
tables and intermediate results, and $0$ for the rest of the nodes in the graph.
Each server is assigned a 2-dimensional capacity
$(s_j, r_j)$, in which $s_j$ is
the total available storage of server $S_j$ as before, and $r_j$ corresponds to the available execution capacity.
Then, once again we can use existing graph
partitioning algorithms and libraries (e.g., METIS \cite{karypis95metis}) to partition the
graph, by optimizing both storage and load balancing constraints simultaneously

\subsection{Replication}
\label{sec.replication}
Replication is important in many application settings for fault tolerance and load balancing.
Existing systems consider replication of tables in which a predetermined number of
replicas for each table is desired (e.g., the common three-replica scheme used in modern
distributed object stores). We propose two heuristics along these lines.
Both heuristics make the assumption that the capacities of all servers $S_1, \ldots, S_l$ are equal to $s$.

Given that {\sc Data Placement} is NP-Hard, it is easy to show that
{\sc Data Placement With Replication} is also NP-Hard. The problem can be modeled by
an integer linear program whose formulation is nontrivial and is described in the appendix.
Unfortunately, an optimal solution is very expensive to compute, even for small instances
of the problem, as integer linear programs take exponential time in the worst case.
In this section we propose two heuristics for computing good data placements
with replication, that utilize the graph partitioning algorithm presented earlier. Formulating
a direct reduction to graph partitioning is left for future work.
For simplicity we focus on the case in which no views or intermediate results are present.
Extensions are straightforward.

For the first heuristic---Algorithm \ref{alg.heuristic1}---we run {\sc Data Placement}
once to get a placement without replication, assuming each
server's capacity is only $\lfloor s/r\rfloor$, where $r$ is the desired replication factor. We assume this capacity is enough to
store all the tables.
We then apply this placement strategy $r$ times to $r$ \emph{random permutations} of the server set.
Notice that this algorithm might result in some
tables' being stored multiple times on the same server (of course, in that case we only keep
a single copy). Clearly, this algorithm has the same complexity as {\sc Data Placement}, ignoring the time to
generate the random permutations, but can
only guarantee at most $r$ (and not exactly $r$) replicas of each table.

\begin{pseudocode}[plain]{Heuristic 1}{T, Q, S, r}
\label{alg.heuristic1}
\text{Assignment } A = \text{{\sc Data Placement}}(T, Q, S) \\
    \FOR p \GETS 1 \TO r \DO
    \BEGIN
        \text{Permutation } P = \text{Permute}(S) \\
		\text{Assign } A \text{ to } P. \\
	\END \\
\end{pseudocode}

In the second heuristic, instead of dividing the space on all servers by $r$ and filling a fraction of each server in every
replication round, we allocate $a = \lfloor l/r \rfloor$ servers completely for each round ($l$ is the total number of servers),
and we use these servers to compute a data placement that optimizes the communication cost for a subset of
$\lfloor m/r \rfloor$ queries ($m$ is the total number of queries).\footnote{The last round may have more servers and
more queries to optimize for.}  We assume that one replica of all tables can fit in $a$ servers (so that $r$ replicas can fit in $l$ servers).
For any nontrivial assignment, we have $r < l$, as otherwise all tables will be placed on a single server resulting in $0$ communication cost.

Algorithm~\ref{alg.heuristic2} implements the second heuristic.  In the first round, we run {\sc Data Placement} using only the first $a$ servers.
We then remove from the query set $Q$ the cheapest $\lfloor m/r \rfloor$ queries, denoted $Q^1$.  These cheapest queries
(i.e., those with the lowest communication cost) will not be considered in the subsequent replication rounds.
In the second round, we run {\sc Data Placement} using the next $a$ servers for the query set $Q \setminus Q^1$.
Then, again, we remove the cheapest $\lfloor m/r \rfloor$ queries from the remaining queries, and so on.  In the algorithm below,
$a_{i-1}=(i-1)\lfloor\frac{l}{r}\rfloor$ for $i=1,2,...,r-1$ and $a_r=l$.

\begin{pseudocode}[plain]{Heuristic 2}{T, Q, S, r}
\label{alg.heuristic2}
\FOR i \GETS 1 \TO r \DO
\BEGIN
 P_i = \{S_{1 + a_{i-1}}, S_{a_i}\} \\
\text{Assignment } A = \text{{\sc Data Placement}}(T, Q, P_i) \\
\text{Assign } A \text{ to } P_i\\
Q = Q \setminus \{ \text{the cheapest} \lfloor m/r \rfloor \text{queries remaining in } Q \}\\
\END \\
\end{pseudocode}

Note that we use the full set of tables $T$ in each replication round, meaning that every table will be replicated exactly $r$ times.  The intuition behind this heuristic is that any optimal solution would in fact partition the queries into $r$ parts, such that the queries in each part use one of the replicas of the tables, and do so optimally.  We are essentially trying to approximate an optimal solution by assuming that the $\lfloor m/r \rfloor$ cheapest queries in each round will be served by the replicas placed in this round (which is why the cheapest queries are removed from consideration after each round).  Interestingly, when the algorithm terminates and returns the placement of all the replicas of each table, a query will be executed on a server which results in the minimum communication cost; therefore, it is not necessarily true that any query $q \in Q_i$ be executed using the $i$th copy of the tables placed in the $i$th round.

\section{Experimental Results} \label{sec.experiments}
For our experimental evaluation, we implemented the proposed integer programs in AMPL and solved them optimally using CPLEX.  We also implemented an approximation algorithm that converts the given workload into a bipartite graph using our proposed reductions and runs METIS on the resulting 
graph.\footnote{For METIS parameters, we use $k$-way partitioning, 1000 cuts, and 1000 iterations. We try  $\approx$ 10 u-factors
that correspond to various maximum part size objectives and take the best result that satisfies the part size objectives (i.e., server capacity constraints).  Since METIS is fast,
the NP-hardness of feasibility implies that 
METIS (like any
other polynomial-time algorithm used for {\sc Data Placement}), must sometimes violate the part-size objectives
(unless P$=$NP).}
All experiments were run on a server with 32 Intel IA-64 1.5MHz
cores and 256GB of main memory, running SUSE Linux 2.6.16.

Many of our experiments involving CPLEX had not finished
running after more than three weeks, so we report the data communication cost of the solution
that CPLEX had converged to at that time.
On the other hand, METIS took no more than four minutes in the worst case, to report a good placement. For very few servers (i.e., 2 to 3) CPLEX finds an optimal solution quickly and is therefore preferable to
METIS (even in these cases METIS is very close to the CPLEX optimal).  Hence we do not consider small numbers of servers here.

As the input dataset we use the TPC-DS decision support benchmark \cite{tpc-ds}
(we only need the schema, the queries, and the table sizes,
but we do not need to actually generate any data).
The benchmark contains seven fact tables and 17 dimension
tables, as well as 99 predefined queries. 
The minimum, maximum and average number of table dependencies
per query are 1, 13, and 4, respectively. In other words, the queries are fairly well distributed across the spectrum
from simple to very complex.
As for the table sizes, we set the fact tables to be one to two orders of magnitude larger than the dimension tables and normalized the sizes.
In particular, fact tables have sizes between 50 and 100, and dimension tables have sizes between 1 and 10.

First we test our simple bipartite graph construction with tables and queries only. Only the base
tables need to be stored. When the queries are executed, they are executed at the server specified by the
placement given by the algorithm and the results are then discarded. The total communication cost for
executing a workload is the sum of the costs for executing each query, and the cost of executing a query
is the sum of the sizes of the tables that need to be transferred to the server executing the query. For
simplicity, we assume here that when a table is involved in a particular query,
that the whole table needs to be transferred, if it is not already co-located with the query.
Next, we randomly assign sizes to all queries (between 1 and 20) and assume that all 99 queries are to be
materialized. Then we test the bipartite graph construction for arbitrary execution plans. We also perform experiments to show the effect of
load balancing on the data communication cost. Finally, we
run experiments for testing the replication heuristics.

Given that the total numbers of tables and queries in TPC-DS are fixed, we created very large randomized datasets
for testing the scalability of our solution.  These randomized datasets range from 1000 tables and 1000
queries up to 16000 tables and 16000 queries. We choose the table sizes randomly from a normal
distribution with mean 10 and standard deviation 15 (ignoring negative values and taking the floor of each generated value).
We similarly choose the number of tables per query from
a normal distribution with mean 5 and standard deviation 3. Finally, the tables involved in each query are chosen
uniformly at random.

In our experiments we vary the number of servers from 4 up to 16, and we also vary the capacity of each
server. Notice that for a placement to be feasible, there is a minimum required capacity per server, that
depends both on the total size of all tables and on the largest table.  Given our normalized table sizes
described above, the TPC-DS dataset (with no materialized queries)
has a total size of 580 units of space and there exist three tables with size 100 and four tables with size at least 50.
Hence, with 4 servers, if each server had capacity $580/4=145$, there would not exist any feasible
placement, given that after placing all tables with size 100, there is not enough space left to place all
tables with size 50. In this case, we assign each server a capacity of 150. Similarly, assuming a total
of 16 servers, each server must have capacity $\lceil 580/16\rceil =37$, but clearly now the
largest table cannot fit in a single server. Thus, our algorithms assume that the capacity of
every server is at least 100. In this case the servers are over-provisioned, which means that many
servers might be left under-utilized. In other words, adding more servers does not necessarily
result in a more distributed solution. In fact, given this particular dataset, for 8 or more servers the
optimal data placement is always the same and only uses only 8 servers. Determining the optimal
number of servers, server capacity, horizontal partitioning of tables, etc., is dataset and application dependent,
and beyond the scope of this paper. Thus we do not consider it any further.

\subsection{Simple Query Execution Plans (24 Base Tables and 99 Queries)}

Figures \ref{fig.tpc-ds2_noviews_4} and \ref{fig.tpc-ds2_noviews_8}
show the total cost of each placement given by the optimal algorithm (using CPLEX) and the approximate algorithm (using METIS),
for 4 and 8 servers, respectively, (results with more than 8 servers are identical to the 8 server case) and varying server capacities.
In the graphs we also plot the communication cost obtained by allowing METIS to use 10 units of space more
capacity per server than CPLEX (denoted METIS+10) for easier comparison. We can clearly
see from the figures that the cost of the approximate solution is
very close to that of the best solution which CPLEX had achieved
within three weeks (the only runs which completed
within three weeks were those with 4 servers), for a
minuscule fraction of the run time. METIS took under 3 seconds to report a placement, in the worst case. In the
graphs we also plot the minimum, maximum and average part size
assigned to each server (i.e., the total size of the tables assigned to each server).
We can see that for a large number of servers, many servers remain empty, as expected.
Another interesting observation is that by allowing METIS a modest extra 10 units of space
per server over CPLEX, it can consistently ``beat'' CPLEX for all numbers of servers.

\begin{figure}[t]
\centering
\includegraphics[width=3in]{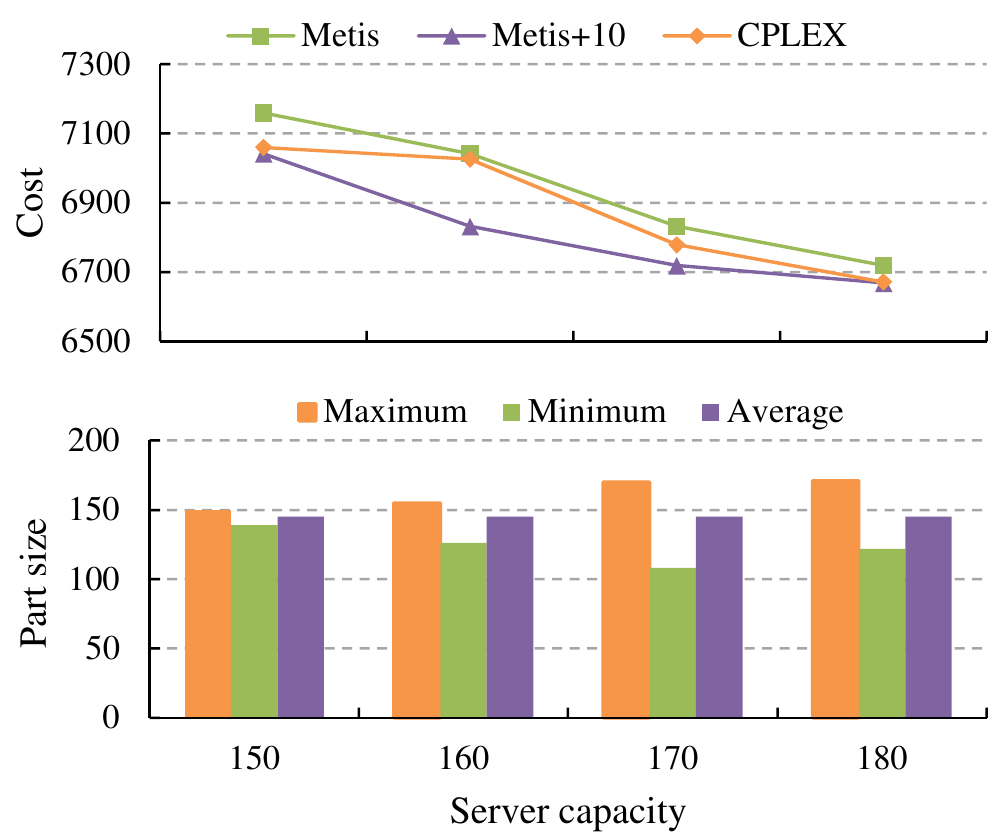}
\caption{No materialized queries. Four servers. Data communication cost.}
\label{fig.tpc-ds2_noviews_4}
\end{figure}

\begin{figure}[t]
\centering
\includegraphics[width=3in]{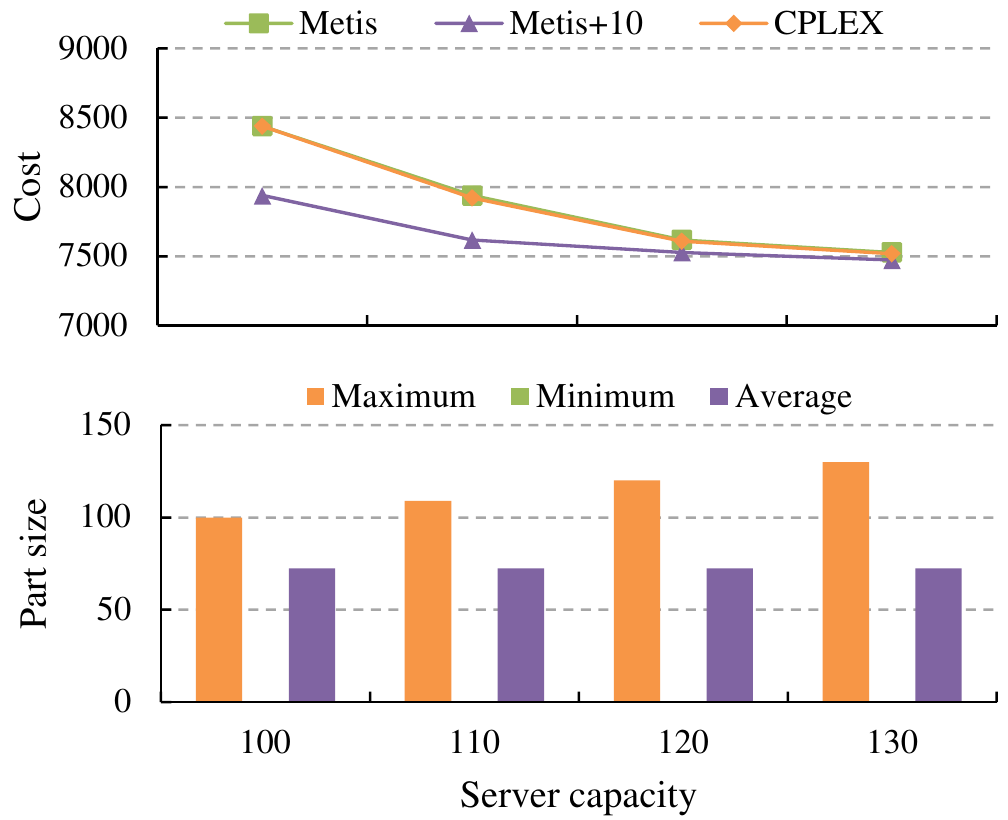}
\caption{No materialized queries. Eight servers. Data communication cost.}
\label{fig.tpc-ds2_noviews_8}
\end{figure}

\subsection{Arbitrary Query Execution Plans (All 99 Queries Materialized)}

In order to test the bipartite graph construction for arbitrary execution plans, we randomly assign
sizes to each one of the $99$ queries in TPC-DS and materialize them. The total table plus materialized view size
of this dataset is 965. Figures \ref{fig.tpc-ds2_views_onemove_4}, \ref{fig.tpc-ds2_views_onemove_8},
and \ref{fig.tpc-ds2_views_onemove_16} show the results assuming 4, 8 and 16 servers, respectively.
Once again, METIS produces very good placements as the number of servers increases.
Furthermore, for 4 servers, METIS is not able to give particularly good solutions.
But even in this case it took CPLEX 220 minutes to find an optimal solution, while METIS
reported the placement in under 3 seconds.

\begin{figure}[t]
\centering
\includegraphics[width=3in]{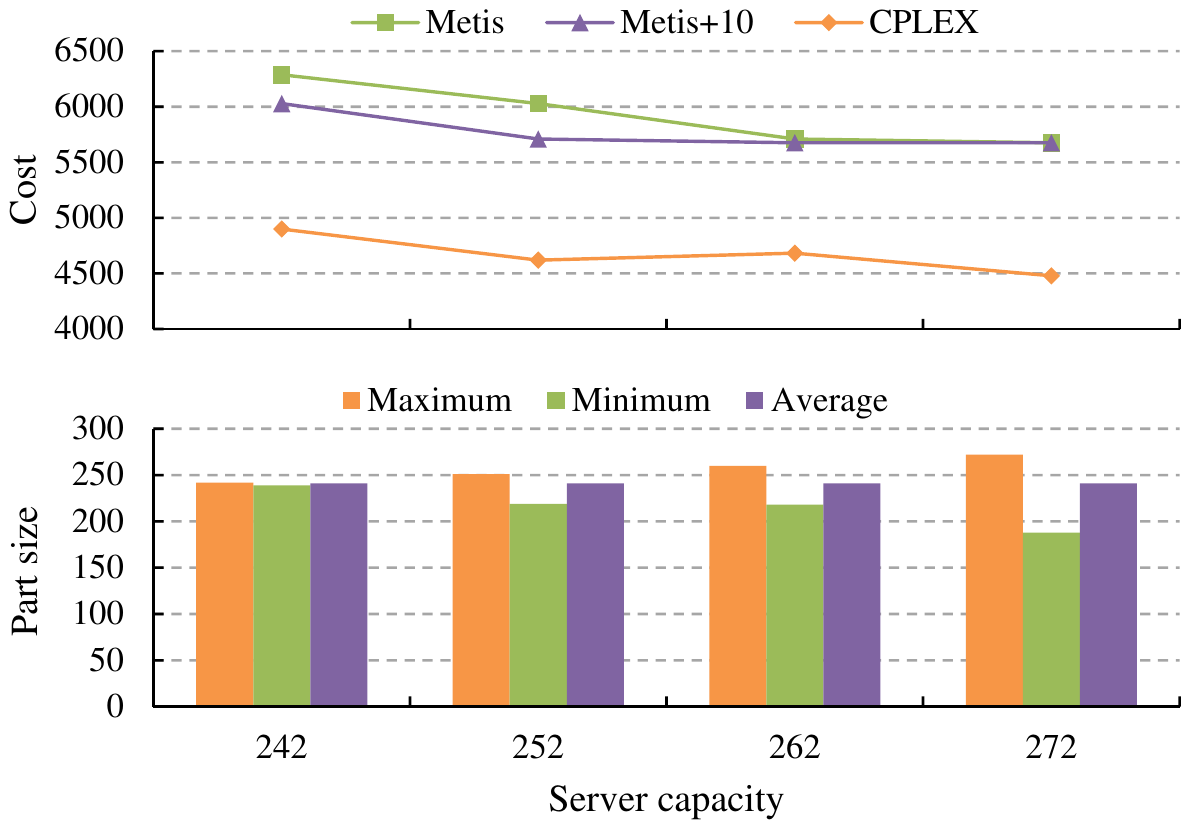}
\caption{Materialized queries. Four servers. Data communication cost.}
\label{fig.tpc-ds2_views_onemove_4}
\end{figure}

\begin{figure}[t]
\centering
\includegraphics[width=3in]{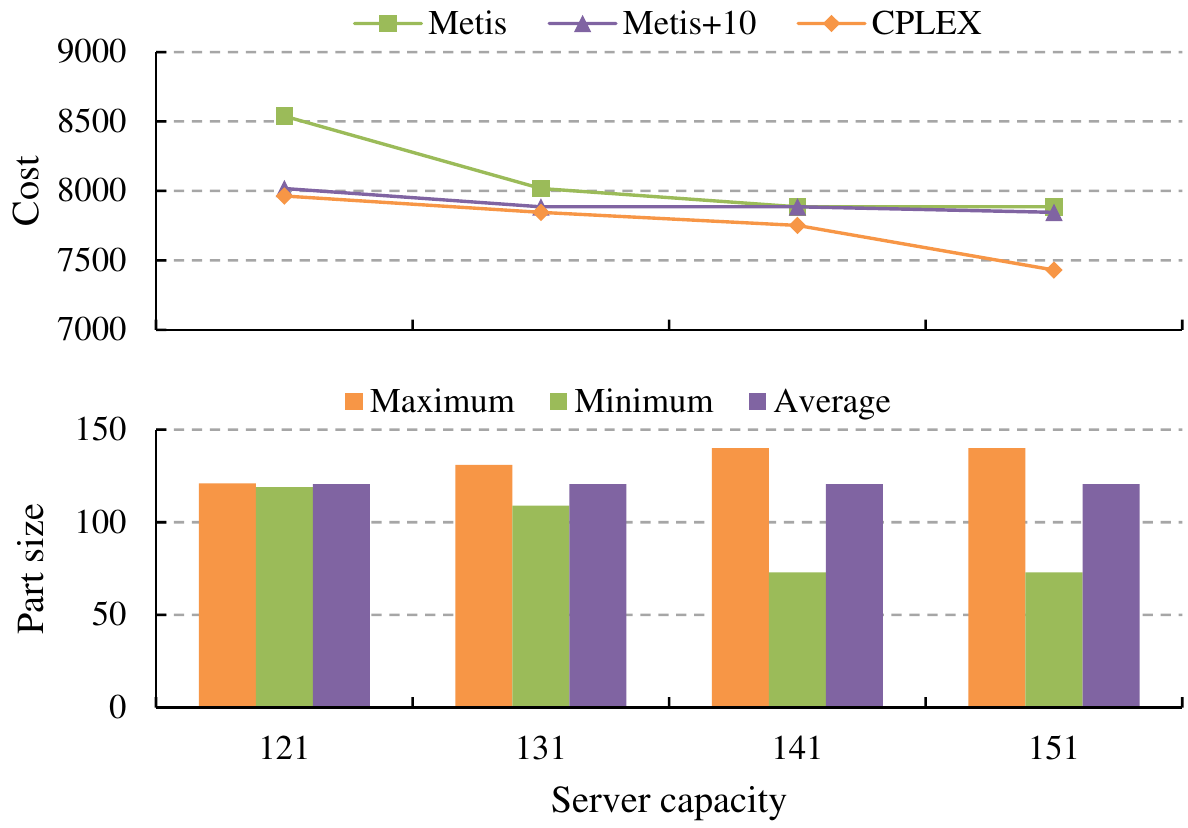}
\caption{Materialized queries. Eight servers. Data communication cost.}
\label{fig.tpc-ds2_views_onemove_8}
\end{figure}

\begin{figure}[t]
\centering
\includegraphics[width=3in]{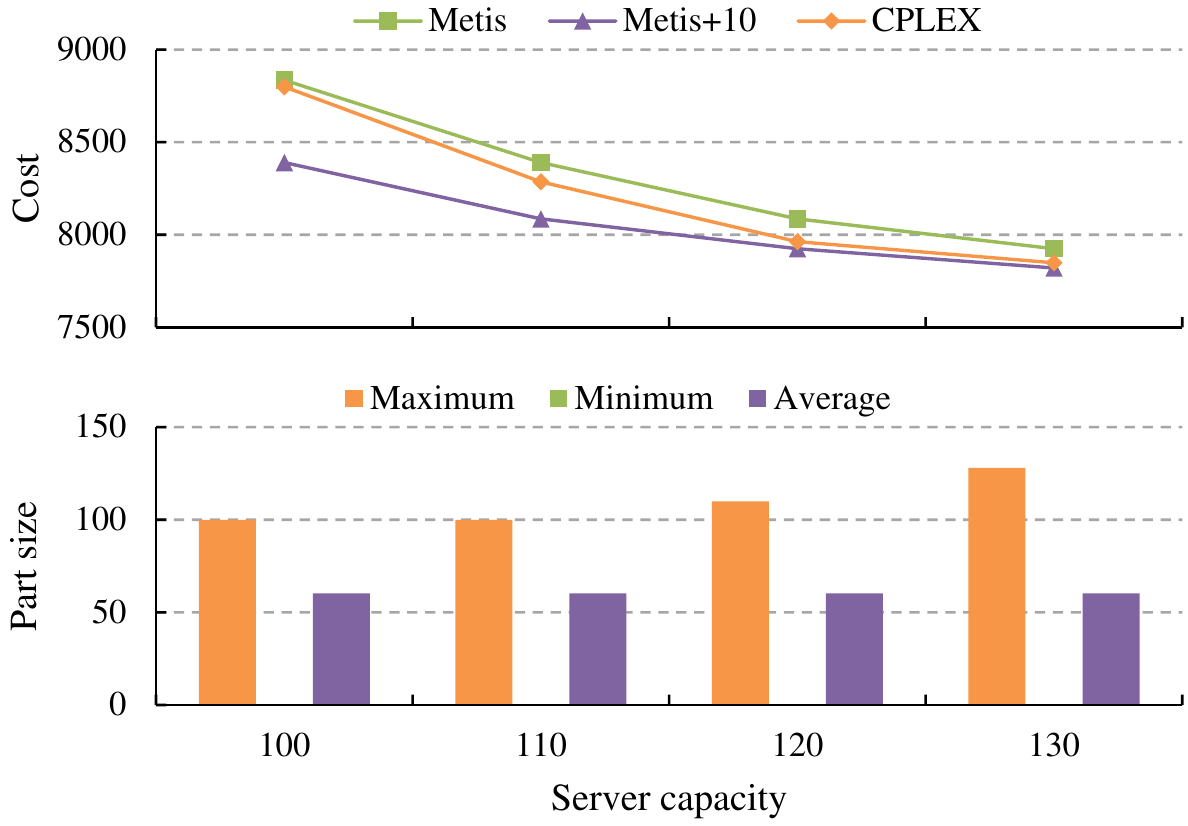}
\caption{Materialized queries. Sixteen servers. Data communication cost.}
\label{fig.tpc-ds2_views_onemove_16}
\end{figure}

For comparison, we also run experiments for generalized data placement in
which we force the materialized views to be
executed and stored on the same server (essentially by setting the appropriate edges to infinity). We wanted
to quantify the impact of allowing the views to move. The results for 8 servers are shown in Figure \ref{fig.onemove-nomove}, with ``-nomove'' representing using the same server to store and compute a view.
It is clear that allowing a view to be computed on one server and stored on another has a noticeable impact on the data communication cost.

\begin{figure}[t]
\centering
\includegraphics[width=3in]{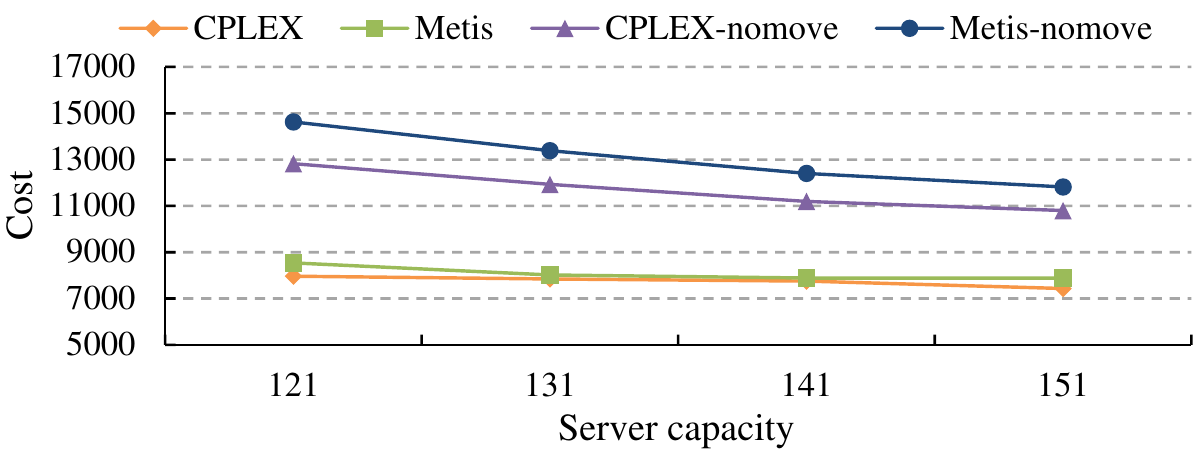}
\caption{Materialized queries. Eight servers. Comparison of data communication cost when allowing views to move vs. not move.}
\label{fig.onemove-nomove}
\end{figure}

\subsection{Load-Balancing Experiments}

\begin{figure}[t]
\centering
\includegraphics[width=3.5in]{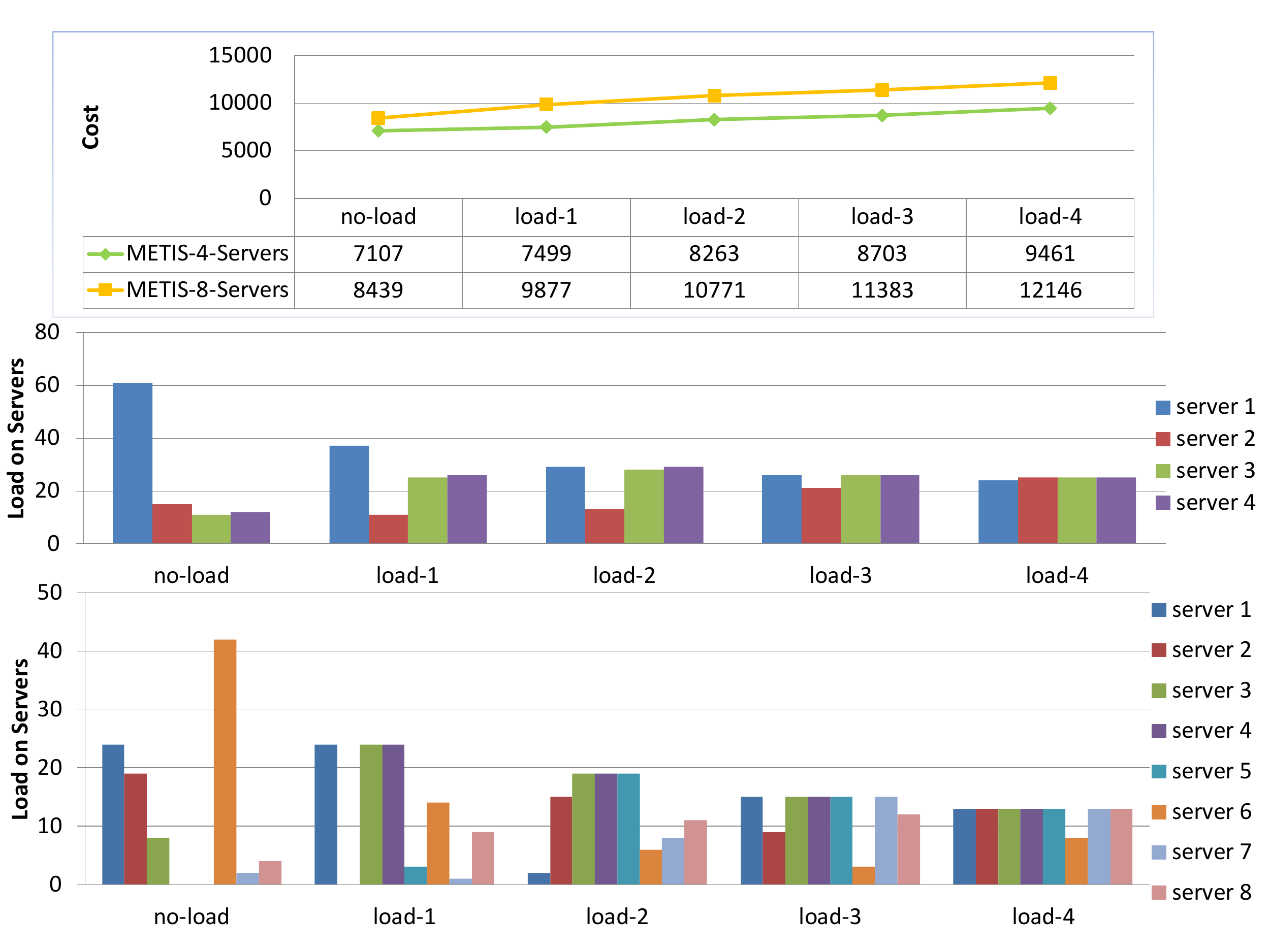}
\caption{No materialized queries. Four and Eight servers respectively. Comparison of data communication cost with distribution of load.}
\label{fig.load}
\end{figure}

We next run a set of experiments to test the effect of adding load-balancing constraints, as proposed in Section~4.3, on the
data communication cost. The bipartite graph construction now contains two-dimensional weights for all vertices and each server has two-dimensional weights: capacity
and load. METIS allows two-dimensional weights on vertices and, using the {\sf ubvec} parameter, it is possible to vary the ratio of minimum-to-maximum
load across the servers. Figure \ref{fig.load} shows a comparison result of using 4 and 8 servers with no load balancing constraints ({\sf no-load}) and four
other scenarios ({\sf load-1} through {\sf load-4})
with increasingly stringent minimum-to-maximum load constraints.
For this experiment, we use the simple TPC-DS workload with base tables and queries only.
The distribution of load across the servers is also shown. The results indicate that with moderate
increase in communication cost it is possible to satisfy reasonable load balancing requirements.
We obtained similar results when all 99 TPC-DS queries were materialized.

\subsection{Replication Experiments}

We also run a set of experiments to test the replication heuristics proposed in Section \ref{sec.replication}.
For these experiments we vary the number of servers from 4 up to 16, and the number of replicas per table
from 1 to 5. Then, we compute the total cost of executing the queries after determining the placement using
Heuristic~1 (abbreviated H1) and Heuristic~2 (abbreviated H2). We use the simple TPC-DS dataset with no materialized views. Figure \ref{fig.replication_cost}
shows the results for 4, 8, and 16 servers. We can clearly see that H2 results in significantly better
data placement than H1 in all cases. Also, notice how the data communication cost drops as the number of
replicas increases, as expected. Observe that as we increase the number of servers, the total communication
cost increases. This is reasonable, given that we do not keep the capacities of the servers fixed (so as not to have
underutilized servers).  Therefore, as the number of servers increases, the total capacity of each server decreases proportionately, and
hence the dataset becomes more distributed, resulting in higher communication cost. Finally, notice that not all
parameter combinations are meaningful, hence there are some missing points in the plots (for example it is pointless to have replication
more than 4 for 4 servers).
Specifically, for H2 certain configurations result in invalid inputs, given that H2 utilizes $1/r$ of the total number of
servers in each round.

Figure \ref{fig.replication_size} shows the maximum part size assigned to the servers as the replication factor increases.
We can see that both heuristics result in fairly equal maximum part size, although H2 is consistently better
than H1. As a baseline, we also plot the ``desired'' part size, which is the ideal storage overhead on each
server, based on the desired replication factor. The desired part size is usually unattainable, since there exist tables
in the dataset with sizes larger than the desired part size.

\begin{figure}[t]
\centering
\includegraphics[width=3in]{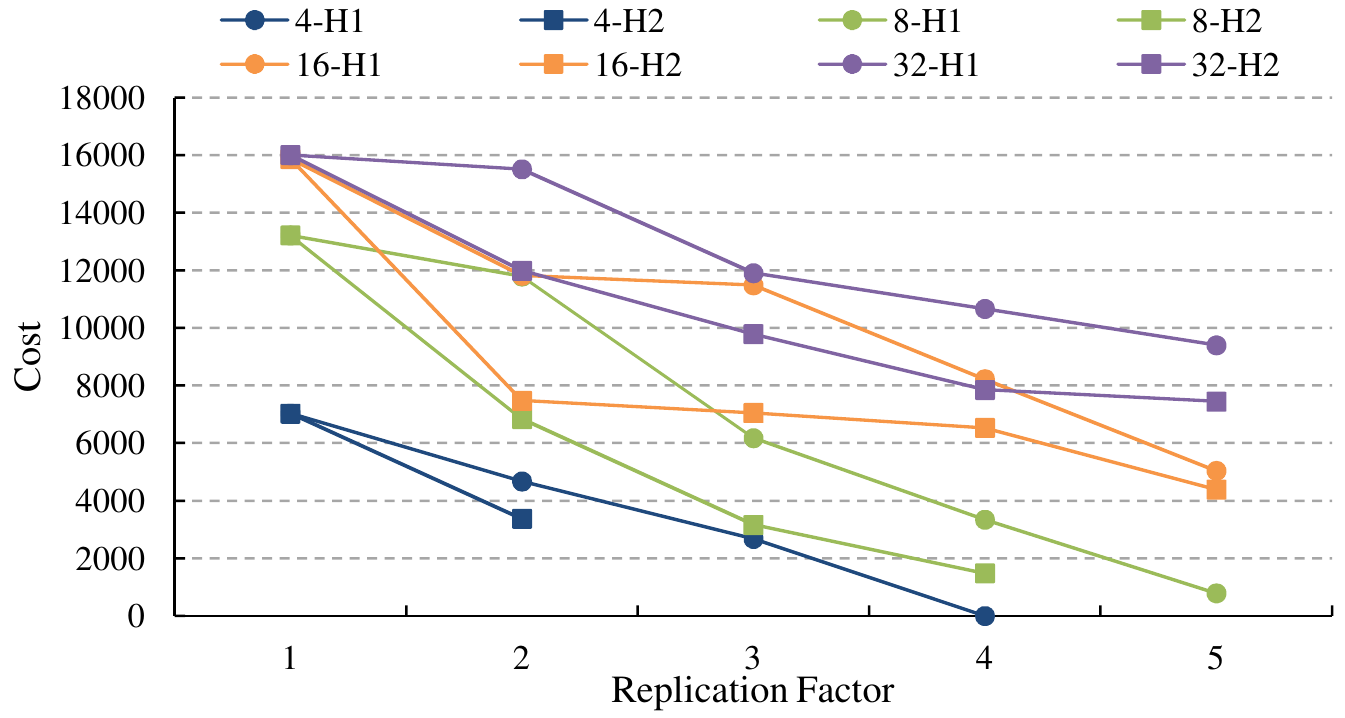}
\caption{No materialized queries. Data communication cost with replication.}
\label{fig.replication_cost}
\end{figure}

\begin{figure}[t]
\centering
\includegraphics[width=3in]{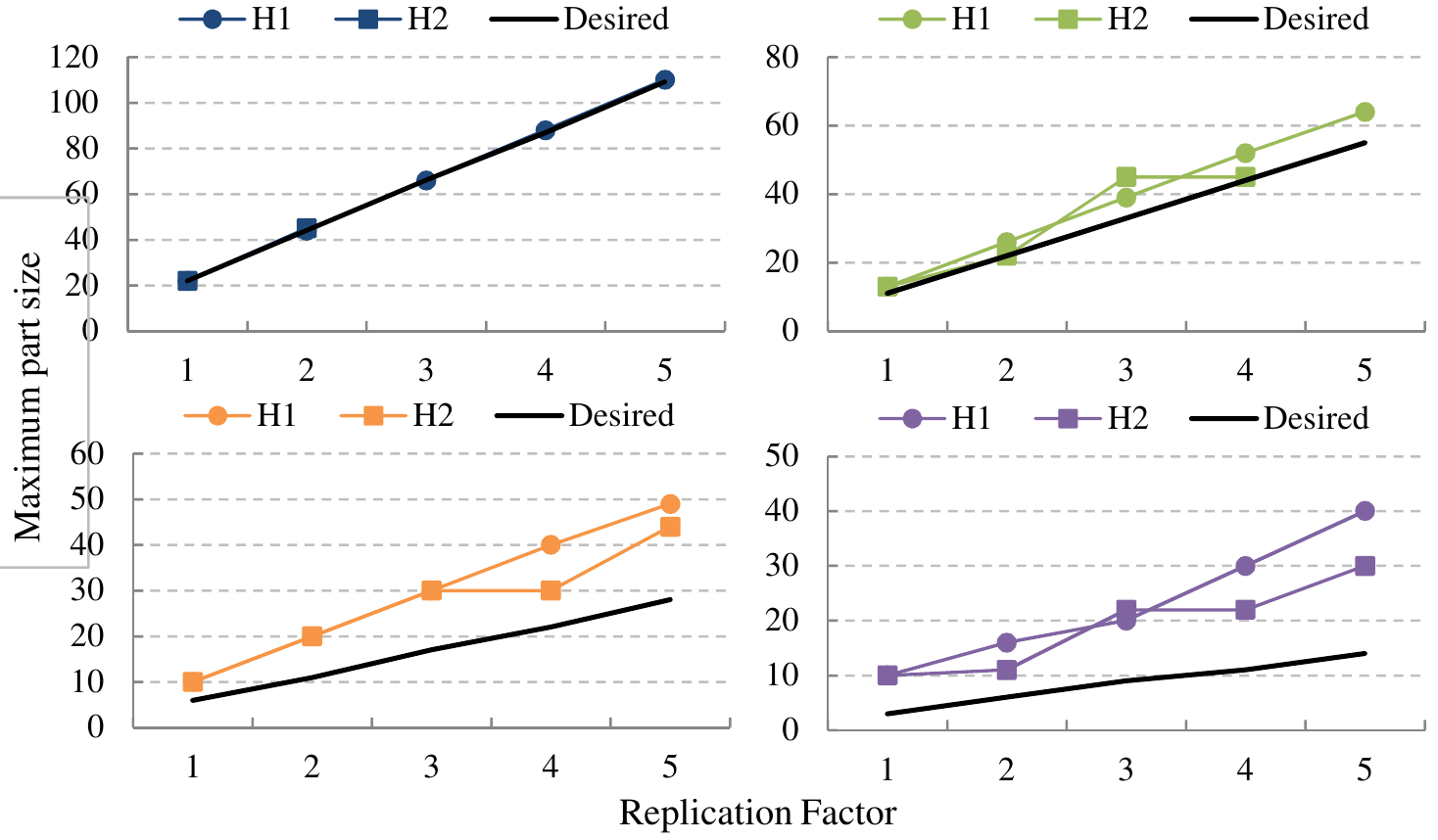}
\caption{No materialized queries. Maximum part size with replication.}
\label{fig.replication_size}
\end{figure}

\subsection{Scalability Experiments}

For our final experiment we test the scalability of METIS using very large randomly-generated workloads.
Figures \ref{fig.rand_queries} and \ref{fig.rand_tables} show the results for a varying number
of queries and tables. In the first plot we keep the number of tables fixed at 1000 and vary
the number of queries from 1000 up to 16000. In the second plot we keep the number of queries
fixed to 8000 and vary the number of tables from 1000 up to 16000. We also vary the number of
servers from 4 up to 128. Notice that the more tables
we have, the sparser the queries become, resulting in sparser bipartite graphs, which are easier
to handle. But the more queries we have with respect to tables the denser the graphs become,
which is expected to result in higher cost. The reported times are in seconds, and we can see that METIS
scales very well across all dimensions.

\begin{figure}[t]
\centering
\includegraphics[width=3in]{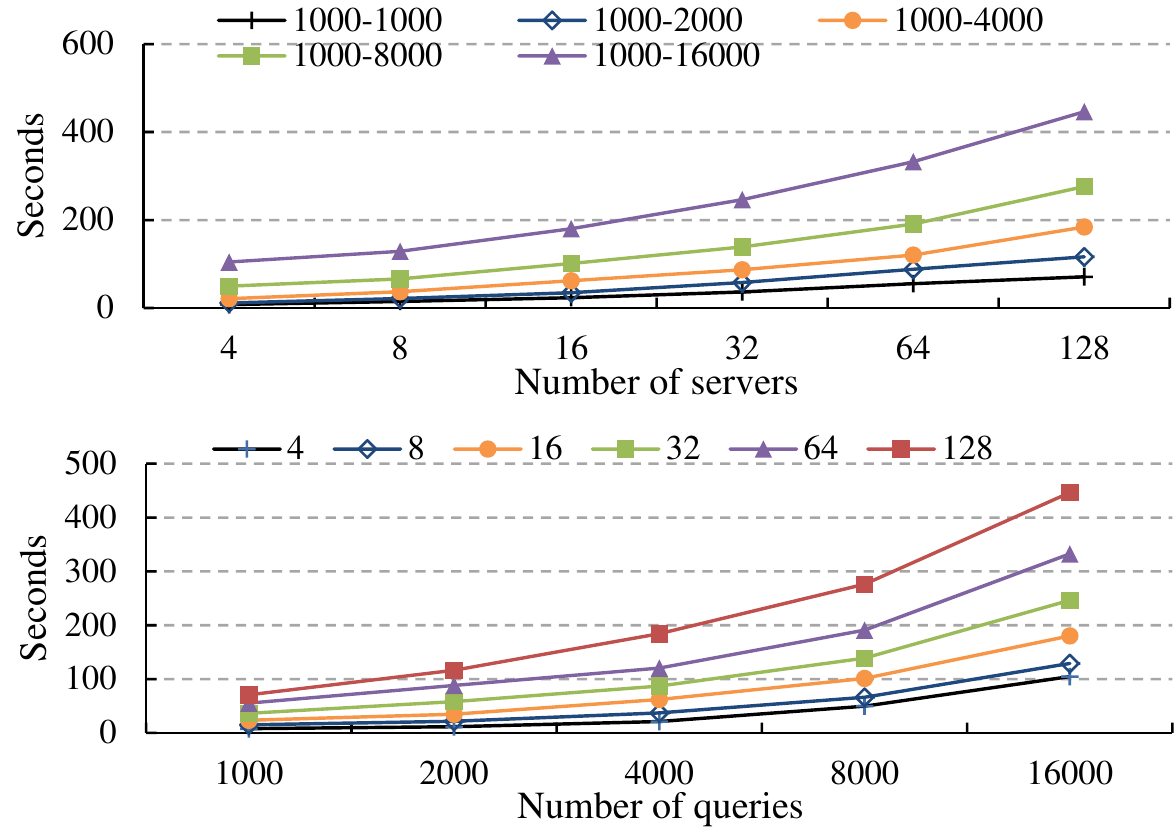}
\caption{Randomly generated datasets. Scalability with varying number of queries.}
\label{fig.rand_queries}
\end{figure}

\begin{figure}[t]
\centering
\includegraphics[width=3in]{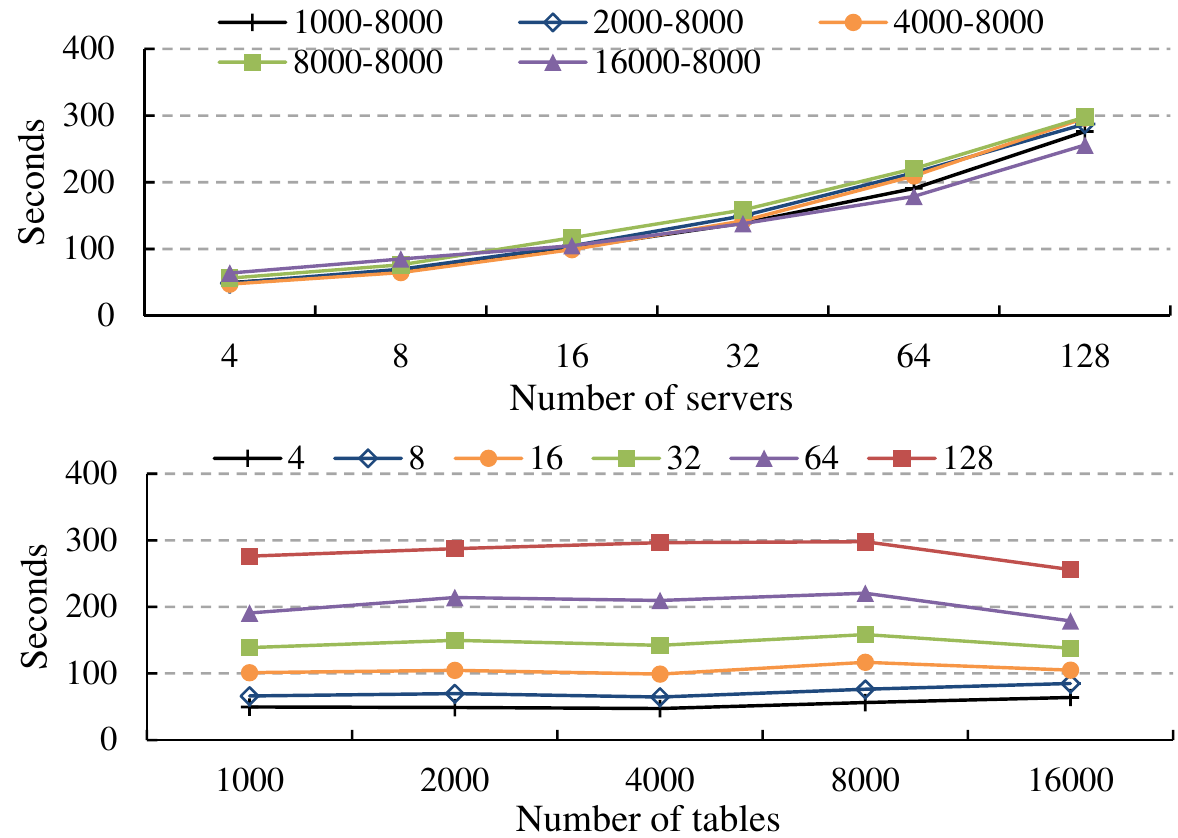}
\caption{Randomly generated datasets. Scalability with varying number of tables.}
\label{fig.rand_tables}
\end{figure}

\subsection{Lessons Learned}

Clearly, relying on an optimal algorithm for determining data placement is not practical even
for very small datasets and number of servers (e.g., 24 tables, 99 materialized queries, and four servers
takes 220 minutes to compute). On the other hand, determining approximate placements using our technique
can be done very quickly and results in very high quality results in most cases, especially as the number
of servers increases. In other words, the reduction of data placement to graph partitioning by means of our
bipartite graph construction allowed us to use standard graph partitioning libraries (METIS)
to solve a computationally expensive problem efficiently and with high quality in practice.

\section{Conclusions} \label{sec.conclusion}
In this paper, we presented \emph{practical} algorithms for minimizing the data communication cost of evaluating a
query workload in a distributed setting.  We reduced the data communication problem to graph partitioning and showed
how to treat arbitrary query execution plans that can involve tables, materialized views, and intermediate results. We also
discussed load balancing, and presented two heuristics for handling replication.

We implicitly assumed that queries and view updates arrive asynchronously.
An interesting direction for future work is to optimize for workloads in which some queries are always executed together, in some pre-specified order.  This complicates the data placement
problem, but may enable solutions with lower data communication costs.  For example, if two views are
maintained on the same server and require the same table to be shipped from some other server, then by
always updating these two views together we only need to ship said table once per update.
Load balancing may also be improved by knowing which queries tend to be executed together.  For instance, if queries $Q_1$ and $Q_2$ are always executed together and each consists of two steps, then we may prefer to assign their first steps to two different servers, rather than allocating the first steps of both queries to the same server and having the servers allocated to the second steps sit idle until the first steps are done.

\section{Appendix}
\subsection{Hardness}
\begin{theorem}
Determining feasibility of {\sc Data Placement} instances is NP-Hard.
\end{theorem}
\proof
We reduce {\sc Partition}, which is known to be NP-Hard, to (the
feasibility version of) {\sc Data Placement}.  {\sc Partition} is this problem:  given a
sequence $\langle a_1,a_2,\ldots,a_n\rangle$ of positive integers,
is there a subset $S\subseteq \{1,2,\ldots,n\}$ such that
$\sum_{i\in S} a_i=\sum_{i\not \in S} a_i$ ($=(1/2)\sum_i a_i$)?
Given an instance $\langle a_1,a_2,\ldots,a_n\rangle$ of positive
integers, let $Y=\sum_i a_i$.  If $Y$ is odd, then the answer to
{\sc Partition} is ``no.''  If $Y$ is even, construct an
instance of {\sc Data Placement} in which there are two servers,
each of capacity $Y/2$.  There are $n$ tables, the $i$th of
which has size $a_i$, and there are no queries.  The key point
is that there is a legal placement of the $n$ tables onto the
two servers, each of capacity $Y/2$, if and only if the answer
for {\sc Partition} is ``yes.'' \qed

\subsection{Integer Program (IP) For the Case of Replication}
The IP assumes that we want $r \ge 1$ replicas of each
table (though generalizing to the case of table-dependent
replica counts is easy).
The $r$ replicas of table $T_j$ will be denoted by superscript $h=1,2,...,r$.
The IP for this case has binary variables $x^h_{T_j,S_k}$ for table
$T_j$ and server $S_k$.  The variable $x^h_{T_j,S_k}$ will be 1 if and only if the
$h$th replica of table $T_j$ is stored on server $S_k$.
In addition, for query $Q_i$ and server $S_k$, there is a binary variable
$y_{Q_i,S_k}$, which is 1 if and only if $Q_i$ is stored on
server $S_k$.
Finally, we also have real (not binary) variables
$z^h_{Q_i,T_j,S_k}$ for query $Q_i$, table
$T_j$ in $Q_i$, and server $S_k$, where $z^h_{Q_i,T_j,S_k}$ will
be 1 if the $h$th replica of table $T_j$ is stored on server
$S_k$ and query $Q_i$ is also stored on server $S_k$,
and 0 otherwise.  In other words,
$z^h_{Q_i,T_j,S_k}$ will equal $y_{Q_i,S_k}\cdot x^h_{T_j,S_k}$.  As product is
decidedly nonlinear, we will have to show how to achieve
$z^h_{Q_i,T_j,S_k}=y_{Q_i,S_k}\cdot x^h_{T_j,S_k}$, since we cannot write such a
constraint explicitly.

Recall that the objective function is to minimize
the sum, over $Q_i$, of the sum of the sizes of all tables in
$Q_i$ minus the sum of the sizes of the tables of
$Q_i$ at least one replica of which is stored on the same server
as $Q_i$.
Hence, we can view the objective as the {\em maximization} of [the sum, over
$Q_i$, of the sizes of the items of $Q_i$ at least one replica of
which is stored on the same server as $Q_i$].
The latter quantity equals
$$=\sum_{k} \sum_{i} y_{Q_i,S_k}\left [ \sum_{j: T_j \in Q_i} C^{T_j}_{Q_i} \left (
\sum_h x^h_{{T_j},S_k}\right ) \right ].$$
(We ensure below, for all $T_j$ and $S_k$, that $\sum_h
x^h_{T_j,S_k}\le 1$,
so that the $r$ replicas of a given table go on
different servers.)
Hence the objective function is

$$\max \sum_{k} \sum_{i} y_{Q_i,S_k} \left [
\sum_{j: T_j\in Q_i} C^{T_j}_{Q_i} \left ( \sum_h x^h_{{T_j}S_k}\right ) \right ].$$

This objective is, unfortunately, quadratic.  We will fix this shortly.
The constraints are:

\begin{enumerate}
\item
For all $Q_i$, $\sum_{k} y_{Q_i,{S_k}}=1$, i.e., every query gets assigned
to exactly one server.
\item
For all ${T_j},{S_k}$, $\sum_h x^h_{{T_j},{S_k}}\le 1$, i.e., at most one replica of
each item goes on a given server.
\item
For all ${T_j},h$, $\sum_{k} x^h_{T_j,{S_k}}=1$, i.e., every replica gets
assigned to some server.
\item
For all ${S_k}$, $\sum_{j} \sum_h x^h_{{T_j},{S_k}} t_j \le s_k$,
i.e., the capacities of the servers are not violated.
\end{enumerate}

We deal with the nonlinearity of the objective function as
follows.  In
$$\max \sum_{k} \sum_{i}
\sum_{j:T_j\in Q_i} C^{T_j}_{Q_i} \left ( \sum_h
x^h_{T_j,{S_k}}\cdot y_{Q_i,{S_k}}\right ) ,$$
we replace ``$x^h_{T_j,{S_k}}\cdot y_{Q_i,{S_k}}$'' by
``$z^h_{Q_i,T_j,{S_k}}$'',
getting $$\max \sum_{k} \sum_{i}
\sum_{j: T_j\in Q_i} C^{T_j}_{Q_i}\left ( \sum_h
z^h_{Q_i,T_j,{S_k}}\right ) .$$
The intent of $z^h_{Q_i,T_j,{S_k}}$ is that it should, at
optimality, equal $x^h_{T_j,{S_k}} \cdot y_{Q_i,{S_k}}$, but how do we ensure
that? The answer is that we add constraints
$z^h_{Q_i,T_j,{S_k}} \le y_{Q_i,{S_k}}$ and  $z^h_{Q_i,T_j,{S_k}} \le
x^h_{T_j,{S_k}}$,
for all $h,Q_i,T_j,{S_k}$.  Because the objective function is a
maximization, and $z^h_{Q_i,T_j,{S_k}}$ appears nowhere else, in an
optimal solution $z^h_{Q_i,T_j,{S_k}}$ will be as large as possible, and
hence at optimality
$z^h_{Q_i,T_j,{S_k}}=\min\{y_{Q_i,{S_k}},x^h_{T_j,{S_k}}\}=
y_{Q_i,{S_k}} \cdot x^h_{T_j,{S_k}}$.


\end{document}